\documentclass[a4paper,11pt]{article}
\usepackage[english]{babel}
\usepackage{amsfonts}
\usepackage[latin1]{inputenc}
\usepackage{array}
\usepackage{color}
\usepackage{hyperref}
\usepackage{subfig}
\usepackage{amsmath}
\usepackage{amsthm}
\usepackage{amsxtra}
\usepackage{amstext}
\usepackage{amssymb}
\usepackage{mathrsfs}
\usepackage{latexsym}
\usepackage{graphicx,wrapfig}
\usepackage[all]{xy}

\newtheorem{teo}{Theorem}[section]
\newtheorem{lemma}[teo]{Lemma}
\newtheorem{prop}[teo]{Proposition}
\newtheorem{coro}[teo]{Corollary}
\newtheorem{esempio}[teo]{Example}
\theoremstyle{definition}
\newtheorem{defn}[teo]{Definition}
\newtheorem{oss}[teo]{Remark}

\newcommand{\deltabr}{\Delta \langle B\rangle}
\newcommand{\sigmau}{\overline{\sigma}^2}
\newcommand{\sigmad}{\underline{\sigma}^2}

\newcommand{\ii}{_{t_i}}
\newcommand{\ipl}{_{t_{i+1}}}
\newcommand{\imi}{_{t_{i-1}}}
\newcommand{\NM}{\mathbb{N}}
\newcommand{\Pro}{\mathcal{P}}

\newcommand{\E}{\mathbb{E}}
\newcommand{\brac}{\langle B\rangle}
\newcommand{\F}{\mathcal{F}}

\newcommand{\Hcal}{\mathcal{H}}
\newcommand{\R}{\mathbb{R}}

\newcommand{\ut}{_{t\in[0,T]}}

  {\left\lbrace\begin{array}{@{}l@{}}}%
  {\end{array}\right.}
\DeclareMathOperator*{\esssup}{ess\,sup}

\makeatletter
\newcommand{\subalign}[1]{%
  \vcenter{%
    \Let@ \restore@math@cr \default@tag
    \baselineskip\fontdimen10 \scriptfont\tw@
    \advance\baselineskip\fontdimen12 \scriptfont\tw@
    \lineskip\thr@@\fontdimen8 \scriptfont\thr@@
    \lineskiplimit\lineskip
    \ialign{\hfil$\m@th\scriptstyle##$&$\m@th\scriptstyle{}##$\crcr
      #1\crcr
    }%
  }
}
\makeatother

\numberwithin{equation}{section}

\newcommand{\condesp}[2][\F_t]       { E_G\left.\left[#2\right|#1\right]}
\newcommand{\esp}[2][]       { E_G\left[#2\right]}

\begin{document}
\author{Francesca Biagini\thanks{Department of Mathematics, Workgroup Financial and Insurance Mathematics, University of Munich (LMU), Theresienstra\ss e 39, 80333 Munich, Germany. Emails: francesca.biagini@math.lmu.de, jacopo.mancin@math.lmu.de, meyer-brandis@math.lmu.de} \textsuperscript{,}\thanks{Secondary affiliation: Department of Mathematics, University of Oslo, Box 
1053, Blindern, 0316, Oslo, Norway} \and Jacopo Mancin\footnotemark[1] \and Thilo Meyer Brandis\footnotemark[1]}  
\title{Robust Mean-Variance Hedging via $G$-Expectation}
\date{}
\maketitle

\begin{abstract}
\noindent
In this paper we study mean-variance hedging under the $G$-expectation framework. Our analysis is carried out by exploiting the $G$-martingale representation theorem and the related probabilistic tools, in a continuous financial market with two assets, where the discounted risky one is modeled as a symmetric $G$-martingale. By tackling progressively larger classes of contingent claims, we are able to explicitly compute the optimal strategy under general assumptions on the form of the contingent claim. 
\end{abstract}

\section{Introduction}
Mean-variance hedging is a classical method in Mathematical Finance for pricing and hedging of contingent claims in incomplete markets. In this paper we consider the mean-variance hedging problem in the $G$-expectation framework in continuous time. Our analysis deeply relies on the \emph{quasi probabilistic} tools provided by the $G$-calculus and thus distinguishes itself from other works on model uncertainty such as the BSDEs approach (see~\cite{nunno} as a reference), the parameter uncertainty setting (see for example~\cite{miss}) or the one period model examined in~\cite{weiwei}.
 
The $G$-expectation space, which represents a generalization of the usual probability space,  was introduced in 2006 by Peng~\cite{Peng:Gexp} for modeling volatility uncertainty and then progressively developed to include most of the classical results of probability theory and stochastic calculus (see~\cite{dhp}, ~\cite{comparison}, ~\cite{osuka}, ~\cite{Peng:dynamicrisk}, ~\cite{Peng:representation} and~\cite{soner:representation} to cite some of them). As a result the $G$-expectation theory has become a very useful framework to cope with volatility ambiguity in finance and many authors have studied some classical problems of stochastic finance, such as no arbitrage conditions, super-replication and optimal control problems in this new setting (see for example~\cite{recursive} and~\cite{Vorbrink}). \\
\noindent
In this context we assume that the discounted risky asset $(X_t)_{t\in[0,T]}$ is a symmetric $G$-martingale (see Definition~\ref{defn:martin}). This means that we consider a financial market that is intrinsically incomplete because of the uncertainty affecting the volatility of $X$. Since perfect replication of a claim $H$ by means of self-financing portfolios will not always be possible, we look for the self-financing strategy which is as close as possible in a quadratic sense to $H$ in a robust way. More precisely we aim at solving the optimal problem 
\begin{equation}\label{numb1}
\inf_{(V_0,\phi)\in\mathbb{R}_+\times \Phi } J_0(V_0,\phi)=\inf_{(V_0,\phi)\in\mathbb{R}_+\times \Phi } \esp{\left(H-V_T(V_0,\phi)\right)^2},
\end{equation}
where $\Phi$ is a space of suitable strategies defined in Definition~\ref{defn:strategia} and $V_T(V_0,\phi)$ stands for the terminal value of the admissible portfolio $(V_0,\phi)$. The objective functional can be interpreted as a stochastic game between the agent and the market, the latter displaying the worst case volatility scenario and the former choosing the best possible strategy.
In the classical setting (see~\cite{Schweizer} for an overview), if the underlying discounted asset is a local martingale, this is equivalent to retrieve the Galtchouk-Kunita-Watanabe decomposition of $H$, i.e.\ to find the projection of $H$ onto the closed space of square integrable stochastic integrals of $X$. In the $G$-expectation framework such result cannot be used. However the structure of $G$-martingales has been clarified in several works such as~\cite{Peng:representation}, ~\cite{soner:representation} and~\cite{song}. \\
We base our analysis on these results and consider $H$ with decomposition~\eqref{star} to solve the robust mean-variance hedging problem. Moreover, in order to guarantee the $M_G^2$-integrability of the optimal hedging strategy (see Section \ref{section2}), the volatility uncertainty setting imposes some additional regularity on $H$ with respect to the classical case, namely $H\in L_G^{2+\epsilon}(\mathcal{F}_T)$ for some $\epsilon>0$ instead of $H\in L_G^{2}(\mathcal{F}_T)$. 

From a technical point of view tackling \eqref{numb1} is very different from solving the classical mean-variance problem in a standard probability setting. In fact the nonlinearity of the model prevents the orthogonality of $B$ and $\langle B\rangle$, namely the $G$-Brownian motion and its quadratic variation (see~\cite{hum}). This in turn limits the possibility to compute explicitly expressions of the type
\[
\esp{\int_0^T\theta_sdB_s\int_0^T\xi_sd\langle B\rangle_s},
\]
for suitable processes $\theta$ and $\xi$, which is a desirable condition when adopting a quadratic criterion.\\
Our main contribution is the explicit computation of the optimal mean-variance hedging portfolio for a wide class of contingent claims. As $L_G^{2+\epsilon}(\mathcal{F}_T)$ is the closure under the $\Vert \cdot\Vert_{2+\epsilon}$-norm of $L_{ip}(\mathcal{F}_T)$, we can focus on claims with martingale decomposition \eqref{star}, where the finite variation part is explicitely characterized. As shown by Theorem \ref{theconv}, given any approximating sequence $(H^n)_{n\in\mathbb{N}}\subseteq L_{ip}(\mathcal{F}_T)$ for $H\in L_G^{2+\epsilon}(\mathcal{F}_T)$, we obtain that the optimal value functions $J_n^\ast$ for $H^n$ converge to the optimal value function $J^\ast$ for $H$.

We first assume $\eta$ to be a continuous process, deterministic or depending only on $\langle B \rangle$. The class of claims admitting this particular decomposition is already wide enough and includes the quadratic polynomials of $B$ and the Lipschitz functions of $\langle B \rangle$. This last result is particularly interesting from a practical perspective as it incorporates a wide class of \emph{volatility derivatives}, such as volatility swaps.\\
For this kind of claims we are able to provide a full description of the optimal portfolio. In the general case obtaining a complete description of the optimal mean-variance strategy is much more involved. We consider the situation in which $\eta$ is a piecewise constant process $\eta_s=\sum_{i=0}^{n-1}\eta\ii\mathbb{I}_{(t_i,t_{i+1}]}(s)$ and outline a stepwise procedure that we solve explicitly for $n=2$. In addition we provide a lower and upper bound for the terminal risk. This limitation is not completely unexpected since it analogously arises also in the classical context of one single prior, where the discounted asset price $(X_t)\ut$ is modeled as a semimartingale. In this case the solution to the mean variance hedging problem is only implicit and described in a \emph{feedback form} (see \cite{pham}) as no orthogonal projection of the claim on the space of the square integrable integrals with respect to $X$ is possible.

The paper is organized as follows. In Section 2 we introduce some fundamental preliminaries on the $G$-expectation theory and also present some new results on stochastic calculus. In Section 3 we describe the market model and we formulate the mean-variance hedging problem. In Section 4 we provide the explicit solution for the optimal mean-variance portfolio for some classes of contingent claims. In Section 5 we provide a lower and upper bound for the optimal terminal risk.

\section{$G$-Setting}
We outline here an introduction to the theory of sublinear expectations, $G$-Brownian motion and the related stochastic calculus. The results from this section can be found in~\cite{dhp}, ~\cite{Peng:dynamicrisk} and~\cite{song}. Moreover we present some new insights concerning the $G$-martingale decomposition and $G$-convex functions, and provide new estimates, see Lemma~\ref{minus}, ~\ref{quadrato} and Section 2.4. 
\subsection{The $G$-Expectation}
Let $\Omega$ be a given set and $\mathcal{H}$ be a vector lattice of real-valued functions defined on $\Omega$ containing $1$. $\mathcal{H}$ is a space of random variables. Assume in addition that if $X_1,\dots,X_n\in \mathcal{H}$, then $\varphi(X_1,\dots,X_n)\in\mathcal{H}$ for any $\varphi\in C_{l,Lip}(\mathbb{R}^n)$, $n\geq1$, where $\varphi\in C_{l,Lip}(\mathbb{R}^n)$ denotes the set  of real-valued functions $\psi$ defined on $\mathbb{R}^n$ such that 
\[
|\psi(x)-\psi(y)|\leq C(1+|x|^k+|y|^k)|x-y|, \quad\forall x,y\in\mathbb{R}^n,
\]
where $k$ is an integer depending on the function $\psi$.
A nonlinear expectation is defined as follows.

\begin{defn}
A nonlinear expectation $\E$ is a functional $\mathcal{H}\mapsto\mathbb{R}$ satisfying the following properties
\begin{enumerate}
\item Monotonicity: If $X,Y\in\mathcal{H}$ and $X\geq Y$ then $\E(X)\geq\E(Y)$.
\item Preserving of constants: $\E(c)=c$.
\item Sub-additivity: 
$$
\E(X+Y)\leq\E(X)+\E(Y), \qquad \forall X,Y\in\mathcal{H}.
$$
\item Positive homogeneity: $\E(\lambda X)=\lambda\E(X)$, $\forall \lambda\geq0$, $X\in\mathcal{H}$.
\item Constant translatability. $\E(X+c)=\E(X)+c$.
\end{enumerate}
The triple $(\Omega, \mathcal{H}, \E)$ is called a sublinear expectation space.
\label{defn:sub}
\end{defn}

\noindent
\begin{defn} 
In a sublinear expectation space $(\Omega, \mathcal{H}, \E)$ a random variable $Y\in\Hcal$ is said to be independent from another random variable $X\in\Hcal$ under $\E$ if for any test function $\psi\in C_{l,Lip}(\mathbb{R}^2)$ we have
$$
\E(\psi(X,Y))=\E(\E(\psi(x,Y))_{X=x}),
$$
where $\psi(x,Y)\in\Hcal$ for every $x\in\mathbb{R}$ as $\psi(x,\cdot)\in C_{l,Lip}(\mathbb{R})$.
\end{defn}

\begin{oss}
Note from the previous definition that in a sublinear expectation space the condition ``$X$ is independent to $Y$" does not automatically imply ``$Y$ is independent to $X$".
\end{oss}

\begin{defn}
Let $X_1$ and $X_2$ be two random variables defined on the sublinear expectation spaces $(\Omega_1, \mathcal{H}_1, \E_1)$ and $(\Omega_2, \mathcal{H}_2, \E_2)$ respectively. They are called identically distributed, denoted by $X_1\sim X_2$, if
$$
\E_1(\psi(X_1))=\E_2(\psi(X_2)), \qquad \forall\psi\in C_{l,Lip}(\mathbb{R}).
$$
We call $\bar{X}$ an independent copy of $X$ if $\bar{X}\sim X$ and $\bar{X}$ is independent from $X$.
\end{defn}

\noindent The $G$-normal distribution in a sublinear expectation space is then defined as follows.
\begin{defn}
A random variable $X$ on a sublinear expectation space $(\Omega, \mathcal{H}, \E)$ is called $G$-normal distributed if for any $a,b\geq0$
$$
aX+b\bar{X} \sim \sqrt{a^2+b^2}X,
$$
where $\bar{X}$ is an independent copy of $X$. The letter $G$ denotes the function 
$$
G(y):=\frac{1}{2}\E(yX^2): \mathbb{R}\mapsto\mathbb{R}.
$$
\end{defn}
\noindent Such $X$ is symmetric, i.e.\ $\E(X)=\E(-X)=0$. In addition we have the following identity
$$
G(y)=\frac{1}{2}\overline{\sigma}^2y^+-\frac{1}{2}\underline{\sigma}^2y^-,
$$
with $\overline{\sigma}^2:=\E(X^2)$ and $\underline{\sigma}^2:=-\E(-X^2)$. We write $X$ is $N(\{0\}\times[\underline{\sigma}^2,\overline{\sigma}^2])$ distributed. 

\begin{defn}
A process $(B_t)_{t\geq0}$ on a sublinear expectation space $(\Omega, \mathcal{H}, \E)$ is called $G$-Brownian motion if the following properties hold true:

\begin{itemize}
\item[(i)] $B_0=0$.
\item[(ii)] For each $t,s\geq0$ the increment $B_{t+s}-B_t$ is $N(\{0\}\times[\underline{\sigma}^2s,\overline{\sigma}^2s])$ distributed and independent from $(B_{t_1},B_{t_2,},\dots,B_{t_n})$ for any $n\in\mathbb{N}$, $0\leq t_1\leq\cdots\leq t_n\leq t$.
\end{itemize}

\end{defn}

\noindent We thus have the same properties as in the classical case, as well as that $(B_{t+t_0}-B_{t_0})_{t\geq0}$ is a $G$-Brownian motion for all $t_0\geq0$.
We now introduce the construction of $G$-expectation and the corresponding $G$-Brownian motion. We fix a time horizon $T>0$ and set $\Omega_T:=C_0([0,T],\mathbb{R})$, the space of all $\mathbb{R}$-valued continuous paths $(\omega_t)_{t\in[0,T]}$ with $\omega_0=0$. Let $B=(B_t)_{t\in[0,T]}$ be the canonical process on $\Omega_T$ defined as $B_t(\omega):=\omega_t$, $t\in[0,T]$.

\noindent We consider the following space of random variables:
$$
L_{ip}(\Omega_T):=\{\varphi(B_{t_1},\cdots,B_{t_n})|n\in\mathbb{N}, t_1,\dots,t_n\in[0,T], \varphi\in C_{l,Lip}(\mathbb{R}^n)\}.
$$

\noindent
The $G$-Brownian motion is constructed on $L_{ip}(\Omega_T)$. For this purpose let $(\xi_i)_{i\in\mathbb{N}}$ be a sequence of random variables on a sublinear expectation space $(\tilde{\Omega},\tilde{\mathcal{H}}, \tilde{E})$ such that $\xi_i$ is $G$-normal distributed and $\xi_{i+1}$ is independent of $(\xi_1,\dots,\xi_{i})$ for each integer $i\geq1$. A sublinear expectation on $L_{ip}(\Omega_T)$ is then constructed by the following procedure: for each $X\in L_{ip}(\Omega_T)$ with $X=\varphi(B_{t_1}-B_{t_0},\cdots,B_{t_n}-B_{t_{n-1}})$ for some $\varphi\in C_{l,Lip}(\mathbb{R}^n)$, $t_1,\dots,t_n\in[0,T]$, set
$$
E_G(\varphi(B_{t_1}-B_{t_0},\cdots,B_{t_n}-B_{t_{n-1}})):=\tilde{E}(\varphi(\sqrt{t_{1}-t_0}\xi_{1},\dots,\sqrt{t_{n}-t_{n-1}}\xi_{n}).
$$
It is then possible to show that $E_G$ consistently defines a sublinear expectation on $L_{ip}(\Omega_T)$ and the canonical process $B$ represents a $G$-Brownian motion (see~\cite{Peng:dynamicrisk}). 

\begin{defn}
The sublinear expectation $E_G: L_{ip}(\Omega_T) \mapsto\mathbb{R}$ defined through the above procedure is called \emph{$G$-expectation}. The canonical process $(B_t)_{t\in[0,T]}$ on such sublinear expectation space $(\Omega_T,L_{ip}(\Omega_T), E_G)$ is a \emph{$G$-Brownian motion}.
\end{defn}

\noindent
The related $G$-conditional expectation of the random variable $X\in L_{ip}(\Omega_T)$ under $\Omega_{t_i}:=C_0([0,t_i],\R)$ is defined by
$$
E_G(\varphi(B_{t_1}-B_{t_0},\cdots,B_{t_n}-B_{t_{n-1}})|\Omega_{t_i}):=\psi(B_{t_1}-B_{t_0},\cdots,B_{t_n}-B_{t_{n-1}}),
$$
where $\psi(x_1,\dots,x_i):=\tilde{E}(\varphi(x_1,\dots,x_i,\sqrt{t_{i+1}-t_i}\xi_{i+1},\dots,\sqrt{t_{n}-t_{n-1}}\xi_{n})$. \\
Let now $\Vert\xi\Vert_p:=(E_G(|\xi|^p))^{\frac{1}{p}}$ for $\xi\in L_{ip}(\Omega_T)$, $p\geq 1$. Then for any $t\in[0,T]$, $E_G(\cdot|\Omega_t)$ can be continuously extended to $L_G^p(\Omega_T)$, the completion of $L_{ip}(\Omega_T)$ under the norm $\Vert\xi\Vert_p$. \noindent The following property is quite useful.
\begin{prop}[Proposition 22 of~\cite{Peng:Gexp}]
\label{prop:useful}
Let $Y\in L_G^1({\Omega_T})$ be such that $E_G(Y)=-E_G(-Y)$. Then we have
$$
E_G(X+Y)=E_G(X)+E_G(Y),\qquad \forall\; X\in L_G^1(\Omega_T).
$$
\end{prop}

\noindent
The $G$-expectation can be seen as a ``worst case expectation". Let $\mathcal{F}=\mathcal{B}(\Omega_T)$ be the Borel $\sigma$-algebra and consider the probability space $(\Omega_T,\mathcal{F},P)$. Let $W=(W_t)_{t\in[0,T]}$ be a classical Brownian motion on this space. The filtration generated by $W$ is denoted by $\mathbb{F}=(\mathcal{F}_t)_{t\in[0,T]}$, where $\mathcal{F}_t:=\sigma\{W_s|0\leq s\leq t\}\vee\mathcal{N}$, and $\mathcal{N}$ denotes the collection of $P$-null subsets. Let $\Theta$ be the bounded closed subset $\Theta:=[\underline{\sigma},\overline{\sigma}]$ such that 
\[
G(y)=\frac{1}{2}\sup_{\sigma\in\Theta}\left(y\sigma^2\right)=
\begin{cases}
\frac{1}{2}y\overline{\sigma}^2  & \text{if $y\geq 0$}, \\
\frac{1}{2}y\underline{\sigma}^2 & \text{if $y<0$},
\end{cases}
\] 
and denote by $\mathcal{A}_{t,T}^\Theta$ the collection of all the $\Theta$-valued $\mathbb{F}$-adapted processes on $[t,T]$. For any $\sigma=(\sigma_t)\ut\in\mathcal{A}_{t,T}^\Theta$ and $s\in[t,T]$ we define
\begin{equation}
\label{equation: law}
B_s^{t,\sigma}:=\int_t^s\sigma_udW_u.
\end{equation}
Let $P^\sigma$ be the law of the process $B_t^{0,\sigma}=\int_0^t\sigma_udW_u$, $t\in[0,T]$, i.e.\ $P^\sigma=P\circ(B^{0,\sigma})^{-1}$. Define 

\begin{equation}
\label{equation:pro1}
\mathcal{P}_1:=\{P^\sigma\;|\;\sigma\in\mathcal{A}_{0,T}^\Theta\},
\end{equation}

\noindent
and $\mathcal{P}:=\bar{\mathcal{P}}_1$, as the closure of $\mathcal{P}_1$ under the topology of weak convergence. We can now formulate the main result (see~\cite{dhp} for the proof):
\begin{teo}
For any $\varphi\in C_{l,Lip}(\mathbb{R}^n)$, $n\in\mathbb{N}$, $0\leq t_1\leq\cdots\leq t_n\leq T$, we have
\begin{equation*}
\begin{split}
E_G(\varphi(B_{t_1},\dots,B_{t_n}-B_{t_{n-1}}))&=\sup_{\sigma\in\mathcal{A}_{0,T}^\Theta} E^P(\varphi(B_{t_1}^{0,\sigma},\dots,B_{t_n}^{t_{n-1},\sigma}))\\
&=\sup_{\sigma\in\mathcal{A}_{0,T}^\Theta} E^{P^\sigma}(\varphi(B_{t_1},\dots,B_{t_n}-B_{t_{n-1}}))\\
&=\sup_{P^\sigma\in\mathcal{P}_1} E^{P^\sigma}(\varphi(B_{t_1},\dots,B_{t_n}-B_{t_{n-1}})).
\end{split}
\end{equation*}
Furthermore,
$$
E_G(X)=\sup_{P\in\mathcal{P}}E^P(X), \qquad \forall X\in L_G^1(\F_T).
$$
\end{teo}

\noindent
Finally, given the set of probability measures $\mathcal{P}$, we introduce here a notation that will be useful later on.
\begin{defn}
A set $A$ is said \emph{polar} if $P(A)=0$ $\forall P\in\mathcal{P}$. A property is said to hold \emph{quasi surely (q.s.)} if it holds outside a polar set.
\end{defn}
\noindent
In the rest of the paper we work in the setting outlined above.

\subsection{Stochastic Calculus of It\^o type with $G$-Brownian Motion}\label{section2}
We now introduce the stochastic integral with respect to a $G$-Brownian motion. To this purpose we summarize some results of~\cite{Peng:dynamicrisk}, if not mentioned otherwise, that are useful in the sequel. For $p\geq1$ fixed, we consider the following type of simple processes: for a given partition $\{t_0,\dots,t_N\}$ of $[0,T]$, $N\in\mathbb{N}$, we set
\begin{equation}
\label{equation:rep}
\eta_t(\omega)=\sum_{j=0}^{N-1}\xi_j(\omega)\mathbb{I}_{[t_j,t_{j+1})}(t),	
\end{equation}

\noindent
where $\xi_i\in L_G^p(\F_{t_i})$, $i\in0,\dots,N-1$. The collection of this type of processes is denoted by $M_G^{p,0}(0,T)$. For each $\eta\in M_G^{p,0}(0,T)$ let $\Vert\eta\Vert_{M_G^p}:=(E_G\int_0^T|\eta_s|^pds)^{\frac{1}{p}}$ and denote by $M_G^p(0,T)$ the completion of $M_G^{p,0}(0,T)$ under the norm $\Vert\cdot\Vert_{M_G^p}$.

\begin{defn}
For $\eta\in M_G^{2,0}(0,T)$ with the representation in~\eqref{equation:rep} we define the integral mapping $I:M_G^{2,0}(0,T)\mapsto L_G^2(\F_T)$ by
$$
I(\eta)=\int_0^T\eta(s)dB_s:=\sum_{j=0}^{N-1}\eta_j(B_{t_{j+1}}-B_{t_j}).
$$
\end{defn}

\begin{lemma}[Lemma 30 of~\cite{Peng:Gexp}]
The mapping $I:M_G^{2,0}(0,T)\mapsto L_G^2(\F_T)$ is a linear continuous mapping and thus can be continuously extended to $I:M_G^{2}(0,T)\mapsto L_G^2(\F_T)$. 
\end{lemma}

\noindent 
It is then possible to show that the integral has similar properties as in the classical It\^o case.

\begin{defn}
The quadratic variation of the $G$-Brownian motion is defined as 
$$
\langle B\rangle_t=B_t^2-2\int_0^t B_sdB_s, \qquad \forall t\leq T,
$$
and it is a continuous increasing process which is absolutely continuous with respect to the Lebesgue measure $dt$ (see Definition 2.2 in~\cite{song}). 
\end{defn}

\noindent
Here $\langle B\rangle_t$, $t\in[0,T]$, perfectly characterizes the part of uncertainty, or ambiguity, of $B$. For $s,t\geq0$, we have that $\langle B\rangle_{s+t}-\langle B\rangle_s$ is independent of $\F_s$ and $\langle B\rangle_{s+t}-\langle B\rangle_s\sim\langle B\rangle_t$. We say that $\langle B\rangle_t$ is $N([\underline{\sigma}^2t,\overline{\sigma}^2t]\times \{0\})$-distributed, i.e., for all $\varphi\in C_{l,Lip}(\mathbb{R})$,
\begin{equation}
\label{maximallydist}
E_G(\varphi(\langle B\rangle_t))=\sup_{\underline{\sigma}^2\leq v\leq\overline{\sigma}^2}\varphi(vt).
\end{equation}
The quadratic variation of the $G$-Brownian motion thus satisfies the following definition.

\begin{defn}
An $n$-dimensional random vector $X$ on a sublinear expectation space $(\Omega,\mathcal{H},\mathbb{E})$ is called \emph{maximally distributed} if there exists a closed set $\Gamma\subset\mathbb{R}^n$ such that
\[
\mathbb{E}(\varphi(X))= \sup_{x\in \Gamma} \varphi(x),
\]
for all $\varphi\in C_{l,Lip}(\mathbb{R}^n)$.
\end{defn}

\noindent
The integral with respect to the quadratic variation of $G$-Brownian motion $\int_0^t\eta_sd\langle B\rangle_s$ is introduced analogously. Firstly for all $\eta\in M_G^{1,0}(0,T)$, and then, again by continuity, for all $\eta\in M_G^{1}(0,T)$.

\begin{defn}
\label{defn:martin}
A process $M=(M_t)_{t\in[0,T]}$, such that $M_t\in L_G^1(\F_t)$ for any $t\in[0,T]$, is called $G$-martingale if $E_G(M_t|\mathcal{F}_s)=M_s$ for all $s\leq t\leq T$. If $M$ and $-M$ are both $G$-martingales, $M$ is called a symmetric $G$-martingale.
\end{defn}

\noindent
Denote, for $t\in[0,T]$ and $P\in\mathcal{P}$,
\[
\mathcal{P}(t,P):=\{P'\in\mathcal{P}:P'=P \text{ on }\mathcal{F}_t\}.
\]
By means of the characterization of the conditional G-expectation (see~\cite{soner:representation} for more details) we have that $M$ is a G-martingale if and only if for all $0\leq s\leq t\leq T$, $P\in\mathcal{P}$,
\begin{equation}
\label{equation:mart}
M_s=\esssup_{Q'\in\mathcal{P}(s,P)}E^{Q'}(M_t|\F_s), \qquad P-a.s.
\end{equation}

\noindent
This shows that a G-martingale $M$ can be seen as a multiple prior martingale which is a supermartingale under each $P\in\Pro$. We next give another characterization of $G$-martingales via the following representation theorem.
\begin{teo}[Theorem 2.2 of~\cite{nonlinear}]
\label{decopeng}
Let $H\in L_{ip}(\Omega_T)$, then for every $0\leq t\leq T$ we have 
\begin{equation}
\label{decopengform}
\condesp[\F_t]{H}= \esp{H}+\int_0^t\theta_sdB_s+\int_0^t\eta_sd\langle B\rangle_s-2\int_0^TG(\eta_s)ds,
\end{equation}
where $(\theta_t)\ut\in M_G^2(0,T)$ and $(\eta_t)\ut\in M^1_G(0,T)$.
\end{teo}
\noindent In particular, the nonsymmetric part
\begin{equation}
\label{equation:k}
-K_t:=\int_0^t\eta_sd\langle B\rangle_s-\int_0^t2G(\eta_s)ds,
\end{equation}
$t\in[0,T]$, is a $G$-martingale that is continuous and non-increasing with quadratic variation equal to zero. A similar decomposition can be obtained for all $G$-martingales in $L_G^\beta(\F_T)$, with $\beta>1$.
\begin{teo}[Theorem 4.5 of~\cite{song}]
\label{teo:decomposition}
Let $\beta>1$ and $H\in L_G^\beta(\F_T)$. Then the $G$-martingale $M$ with $M_t:=E_G(H|\mathcal{F}_t)$, $t\in[0,T]$, has the following representation
\[
M_t= X_0+\int_0^t \theta_s dB_s - K_t,
\]
where $K$ is a continuous, increasing process with $K_0=0$, $K_T\in L_G^\alpha(\F_T)$, $(\theta_t)\ut\in M_G^\alpha(0,T)$, $\forall \alpha\in[1,\beta)$, and $-K$ is a $G$-martingale.
\end{teo}
\noindent It then easily follows as a corollary that a $G$-martingale is symmetric if and only if the process $K$ is equal to zero, thus every symmetric $G$-martingale can be represented as a stochastic integral in the $G$-Brownian motion.

Finally we provide some insights on how the representation of the $G$-martingale $(E_G(H|\mathcal{F}_t))\ut$ is linked to the one of $(E_G(-H|\mathcal{F}_t))\ut$. We focus on the particular class of random variables for which the process $\eta$ appearing in~\eqref{equation:k} is stepwise constant. To ease the notation we explicitly prove the case in which 
\[
\eta_s=\mathbb{I}_{(t,T]}(s)\bar{\eta},
\]
where $0<t<T$, $s\in[0,T]$ and $\bar{\eta}\in L_{ip}(\Omega_t)$, but the generalization to $n$ steps is straightforward. 
\begin{lemma}
\label{minus}
Let 
\[
H=\esp{H}+\int_0^T\theta_s dB_s +\bar{\eta}(\brac_T-\brac_t)-2G(\bar{\eta})(T-t),
\]
where $(\theta_s)_{s\in[0,T]} \in M_G^2(0,T)$, and $\bar{\eta}\in L_{ip}(\F_t)$ is such that
\[
|\bar{\eta}|=\esp{|\bar{\eta}|}+\int_0^t\mu_s dB_s+\int_0^t\xi_s d\brac_s-2\int_0^t G(\xi_s)ds,
\]
for some processes $(\mu_s)_{s\in[0,t]}\in M_G^2(0,t)$ and $(\xi_s)_{s\in[0,t]}\in M_G^1(0,t)$. Then the decomposition of $-H$ is given by
\[
-H=\esp{-H}+\int_0^T\bar{\mu}_s dB_s +\int_0^T \bar{\xi_s}d\brac_s-2\int_0^TG(\bar{\xi}_s)ds,
\]
where 
\[
\bar{\mu}_s=
\begin{cases}
\mu_s(\sigmau-\sigmad)(T-t)-\theta_s,   & \text{if $s\in [0,t]$},\\
-\theta_s, & \text{if $s\in(t,T]$},
\end{cases}
\]
and
\[
\bar{\xi}_s=
\begin{cases}
\xi_s(\sigmau-\sigmad)(T-t),  & \text{if $s\in [0,t]$}, \\
-\bar{\eta}, & \text{if $s\in(t,T]$}.
\end{cases}
\] 
\end{lemma}
\begin{proof}
For $s<t$ we have by the properties of $\langle B\rangle$ and of the conditional $G$-expectation that 
\[
\begin{split}
&\condesp[\F_s]{-H}\\
=&\condesp[\F_s]{-\esp{H}-\int_0^T\theta_u dB_u -\bar{\eta}(\brac_T-\brac_t)+2G(\bar{\eta})(T-t)}\\
=&-\esp{H}-\int_0^s\theta_u dB_u +\condesp[\F_s]{-\bar{\eta}(\brac_T-\brac_t)+2G(\bar{\eta})(T-t)}\\
=&-\esp{H}-\int_0^s\theta_u dB_u +\\
&\quad+\condesp[\F_s]{\condesp[\F_t]{-\bar{\eta}(\brac_T-\brac_t)+2G(\bar{\eta})(T-t)}}\\
=&-\esp{H}-\int_0^s\theta_u dB_u +(\sigmau-\sigmad)(T-t)\condesp[\F_s]{|\bar{\eta}|}\\
=&-\esp{H}+(\sigmau-\sigmad)(T-t)\esp{|\bar{\eta}|}+\int_0^s\left(\mu_u(\sigmau-\sigmad)(T-t)-\theta_u\right) dB_u\\
&\quad + (\sigmau-\sigmad)(T-t)\int_0^s\xi_u d\brac_u-2\int_0^sG(\xi_u(\sigmau-\sigmad)(T-t))du\\
=&\esp{-H}+\int_0^s\left(\mu_u(\sigmau-\sigmad)(T-t)-\theta_u\right) dB_u+\\
&\quad + (\sigmau-\sigmad)(T-t)\int_0^s\xi_u d\brac_u-2\int_0^sG(\xi_u(\sigmau-\sigmad)(T-t))du,
\end{split}
\]
where in the last equality we used the fact that 
\[
\esp{H}+\esp{-H}=\esp{K_T}=\esp{-\bar{\eta}(\brac_T-\brac_t)+2G(\bar{\eta})(T-t)}. 
\]
On the other hand, when $s>t$
\[
\begin{split}
&\condesp[\F_s]{-\bar{\eta}(\brac_T-\brac_t)+2G(\bar{\eta})(T-t)}\\
=&2G(\bar{\eta})(T-t)+\bar{\eta}\brac_t +\condesp[\F_s]{-\bar{\eta}\brac_T}\\
\end{split}
\]
\[
\begin{split}
=&2G(\bar{\eta})(T-t)+\bar{\eta}\brac_t +\bar{\eta}^+\left(\condesp[\F_s]{-\brac_T+\sigmad T}-\sigmad T\right)+\\
&\qquad\qquad\qquad\qquad +\bar{\eta}^-\left(\condesp[\F_s]{\brac_T-\sigmau T}+\sigmau T\right)\\
=&2G(\bar{\eta})(T-t)+\bar{\eta}\brac_t +\bar{\eta}^+(-\brac_s+\sigmad s -\sigmad T)+\bar{\eta}^-(\brac_s-\sigmau s+\sigmau T)\\
=&2G(\bar{\eta})(T-t)+\bar{\eta}\brac_t-\bar{\eta}\brac_s+2G(-\bar{\eta})(T-s)\\
=&2G(\bar{\eta})(T-t)-\bar{\eta}(\brac_s-\brac_t)+2G(-\bar{\eta})(T-t)-2G(-\bar{\eta})(s-t)\\
=&|\bar{\eta}|(\sigmau-\sigmad)(T-t)-\bar{\eta}(\brac_s-\brac_t)-2G(-\bar{\eta})(s-t),
\end{split}
\]
where we used the fact that 
\[
2G(x)+2G(-x)=|x|(\sigmau-\sigmad)\qquad\forall\; x\in\R.
\]
This completes the proof.
\end{proof}

\subsection{$G$-Jensen's Inequality}
Denote now with $\mathbb{S}(d)$ the space of symmetric matrices of dimension $d$. In the framework of $G$-expectation, the usual Jensen's inequality in general does not hold. Nevertheless an analogue to this result can be proved also in this setting, introducing the notion of $G$-convexity. 
\begin{defn} 
A $C^2$-function $h : \mathbb{R}\mapsto\mathbb{R}$ is called $G$-convex if the following condition holds for each $(y, z, A) \in \mathbb{R}^3$:
\[
G(h'(y)A + h''(y)zz^\top ) - h''(y)G(A) \geq 0,
\]
where $h'$ and $h''$ denote the first and the second derivatives of $h$, respectively.
\end{defn}
\noindent
Using this definition, Proposition 5.4.6 of~\cite{Peng:dynamicrisk} shows the following result.
\begin{prop}
\label{jensen}
The following two conditions are equivalent:
\begin{itemize}
\item The function $h$ is $G$-convex.
\item The following Jensen inequality holds: 
\[
\condesp{h(X)}\geq h(\condesp{X}), \quad t\in [0,T],
\]
for each $X\in L_G^1(\F_T)$ such that $h(X)\in L_G^1(\F_T)$.
\end{itemize}
\end{prop}
\noindent
As a particular case we show that the Jensen's inequality holds in the $G$-framework for $h(x)=x^2$, proving that this function is $G$-convex.
\begin{lemma}
\label{quadrato}
In the one dimensional case, the function $x \mapsto x^2$ is $G$-convex.
\end{lemma}
\begin{proof}
According to the definition we have to check if, for each $(y,z,A)\in\mathbb{R}^3$,
\[
G(2yA+2z^2)\geq 2yG(A),
\]
which is 
\begin{equation}
\label{jen}
(yA+z^2)^+\overline{\sigma}^2-(yA+z^2)^-\underline{\sigma}^2\geq y(A^+\overline{\sigma}^2-A^-\underline{\sigma}^2).
\end{equation}
This can be done by cases. When both $A$ and $y$ are greater than zero the condition is obvious. If $A$ is positive but $y$ is negative the only situation to study is when $yA+z^2<0$. In this case Condition~\eqref{jen} becomes
\[
\begin{split}
(yA+z^2)\underline{\sigma}^2&\geq yA\overline{\sigma}^2\\
yA(\underline{\sigma}^2-\overline{\sigma}^2)+z^2\underline{\sigma}^2&\geq0,
\end{split}
\]
which is always satisfied since $yA(\sigmad-\sigmau)>0$. The case in which $A$ is negative is analogue.
\end{proof}

\subsection{Some Estimates}
Motivated by the issues we incurred when dealing with mean-variance hedging, we show here an estimation for the value of $\esp{\int_0^T\theta_tdB_t\int_0^T\eta_td\langle B\rangle_t}$, for suitable processes $(\theta_t)\ut$ and $(\eta_t)\ut$.
\begin{prop}
\label{estimatees}
Let $(\theta_t)\ut$ and $(\eta_t)\ut$ be processes in $M_G^1(0,T)$ such that $(\eta_t\int_0^t\theta_sdB_s)\ut$ and $(\theta_t\int_0^t\eta_sd \langle B\rangle_s)\ut$ both belong to $M_G^1(0,T)$. Then it holds that
\[
\esp{\int_0^T\theta_tdB_t\int_0^T\eta_td\langle B\rangle_t}\leq \esp{\int_0^T2G(\eta_s\int_0^s\theta_u dB_u)ds}.
\]
\end{prop}
\begin{proof}
By applying the It\^o formula for $G$-Brownian motion (see Section 5.4 in~\cite{Peng:Gexp}), we obtain
\[
\begin{split}
&\esp{\int_0^T\theta_tdB_t\int_0^T\eta_td\langle B\rangle_t}\\
=&\esp{\int_0^T \eta_s \left(\int_0^s \theta_u dB_u\right)d\langle B\rangle_s+\int_0^T\theta_s\left(\int_0^s\eta_u d\langle B\rangle_u\right)dB_s}\\
=&\esp{\int_0^T \eta_s \left(\int_0^s \theta_u dB_u\right)d\langle B\rangle_s}.
\end{split}
\]
The result is then achieved by noticing that 
\[
\begin{split}
&\esp{\int_0^T \eta_s \left(\int_0^s \theta_u dB_u\right)d\langle B\rangle_s}\\
=&E_G\Big[\int_0^T \eta_s \left(\int_0^s \theta_u dB_u\right)d\langle B\rangle_s+ \int_0^T2G(\eta_s\int_0^s\theta_u dB_u)ds+\\
&\qquad\qquad\qquad\qquad\qquad - \int_0^T2G(\eta_s\int_0^s\theta_u dB_u)ds\Big]\\
\leq &\esp{\int_0^T2G(\eta_s\int_0^s\theta_u dB_u)ds}.
\end{split}
\]
\end{proof}

\noindent
As a corollary, we apply the result of Proposition~\ref{estimatees} to provide an estimate for the value of $\esp{B_t\langle B\rangle_t}$. Thanks to the It\^o formula for $G$-Brownian motion we see that this problem is equivalent to the computation of $\esp{B_t^3}$. Since
\begin{equation}
\label{estima}
B_t^3=3\int_0^t B_s^2 dB_s + 3 \int_0^t B_s d\langle B\rangle_s
\end{equation}
and 
\begin{equation}
\label{estima2}
B_t\langle B\rangle_t=\int_0^tB_sd\langle B\rangle_s+\int_0^t\langle B\rangle_sdB_s,
\end{equation}
we have 
\[
\esp{B_t^3}=3\esp{\int_0^t B_s d\langle B\rangle_s}=3\esp{B_t\langle B\rangle_t}.
\]
The explicit computation of $\esp{B_t^{2n+1}}$, for $n\in\mathbb{N}$, has been studied extensively in~\cite{hum}, but still no closed form has been retrieved. Hence the following estimates may be of interest.

\begin{coro}
For each $t\in[0,T]$ it holds that
\begin{equation}
\label{cro}
\esp{B_t\langle B\rangle_t}\leq\esp{\int_0^t2G(B_s)ds}=\frac{\sigmau-\sigmad}{\sqrt{2\pi}}\overline{\sigma}\frac{2}{3} t^{3/2}.
\end{equation}
\end{coro}
\begin{proof}
We only have to prove the equality in~\eqref{cro} and to this end we use an approximation argument. To this purpose, let $\{t_i\}_{i=0,\dots,n}$ be a partition of $[0,t]$ with $t_0=0$, $t_n=t$ and $t_i-t_{i-1}=\frac{t}{n}$ for each $i=1,\dots,n$. Then the process $(B_s^n)_{s\in[0,t]}\in M_G^{1,0}(0,T)$ defined as
\[
B^n_t:=\sum_{i=1}^nB_{t_{i-1}}\mathbb{I}_{[t_{i-1},t_i)}(t)
\]
converges in $M_G^1(0,t)$ to $(B_s)_{s\in[0,t]}$. In fact, by direct computation,
\begin{align}
\notag
\esp{|\int_0^t (B_s-B_s^n)ds|}&\leq \esp{\int_0^t |B_s-B_s^n|ds}\leq n\int_0^{t_1}\esp{|B_s|}ds\\
\label{14*}&=n\int_0^{t_1}\frac{\overline{\sigma}\sqrt{2s}}{\sqrt{\pi}}ds=\frac{\overline{\sigma}\sqrt{2}}{\sqrt{\pi}}\frac{2}{3}\left(\frac{t}{n}\right)^{\frac{3}{2}}n\xrightarrow[n\to\infty]{}0,
\end{align}
where we used a result from Example 19 in~\cite{Peng:Gexp} to argue that $\esp{|B_s|}=E[|N(0,\overline{\sigma}^2s)|]$, and the stationarity of the increments of the $G$-Brownian motion. Using this result we can prove that 
\[
\sum_{i=1}^{n} G(B_{t_{i-1}})(t_i-t_{i-1})\xrightarrow[n\to\infty]{L_G^1(0,t)}\int_0^tG(B_s)ds.
\]
In fact
\[
\begin{split}
&\esp{\left|\int_0^tG(B_s)ds-\sum_{i=1}^{n} G(B_{t_{i-1}})(t_i-t_{i-1})\right|}\\
&\esp{\left|\int_0^t((B_s^+\overline{\sigma}^2-B_s^-\underline{\sigma}^2)ds-\sum_{i=1}^{n} (B_{t_{i-1}}^+\overline{\sigma}^2-B_{t_{i-1}}^-\underline{\sigma}^2)(t_i-t_{i-1})\right|}\\
&= \esp{\left|\int_0^t\overline{\sigma}^2(B_s^+-(B_s^n)^+)ds-\int_0^t\underline{\sigma}^2(B_s^--(B_s^n)^-)ds\right|}
\\
&\leq\esp{\left|\int_0^t\overline{\sigma}^2(B_s^+-(B_s^n)^+)ds\right|}+\esp{\left|\int_0^t\underline{\sigma}^2(B_s^--(B_s^n)^-)ds\right|}\\
&\leq\esp{\int_0^t\left|\overline{\sigma}^2(B_s^+-(B_s^n)^+)\right|ds}+\esp{\int_0^t\left|\underline{\sigma}^2(B_s^--(B_s^n)^-)\right|ds},
\end{split}
\]
and the value of both expectations tends to zero as $n$ goes to infinity as $|X^+-Y^+|\leq|X-Y|$, $|X^--Y^-|\leq|X-Y|$ and because of~\eqref{14*}. We now evaluate 

\begin{align}
\notag&\quad\;\esp{\sum_{i=1}^n2G(B\imi)(t_i-t_{i-1})}\\
\label{pro33}&=\esp{\sum_{i=0}^{n-1}2G(B_{{t_i-1}})(t_i-t_{i-1})+\condesp[\F_{t_{n-2}}]{2G(B_{t_{n-1}})(t_n-t_{n-1})}}.
\end{align}
To compute the term inside the conditional expectation in~\eqref{pro33} note that 
\[
\condesp[\F_{t_{n-2}}]{2G(B_{t_{n-1}})(t_n-t_{n-1})}=f(B_{t_{n-2}}),
\]
where
\[
\begin{split}
f(x)&:=\condesp[\F_{t_{n-2}}]{2G(B_{t_{n-1}}-B_{t_{n-2}}+x)(t_n-t_{n-1})}\\
&=\esp{2G(B_{t_{n-1}}-B_{t_{n-2}}+x)(t_n-t_{n-1})}\\
&=E^{P^{\overline{\sigma}}}\left[2G(B_{t_{n-1}}-B_{t_{n-2}}+x)(t_n-t_{n-1})\right]
\end{split}
\]
since $2G(B_{t_{n-1}}-B_{t_{n-2}}+x)(t_n-t_{n-1})$ is a convex function of $B_{t_{n-1}}-B_{t_{n-2}}$. Proceeding by induction and letting $n\to\infty$ we get
\[
\begin{split}
&\esp{\int_0^t2G(B_s)ds}=E^{P^{\overline{\sigma}}}\left[\int_0^t2G(B_s)ds\right]\\
&\;=\int_0^t\left(\sigmau E^P\left[ \left(\int_0^s \overline{\sigma}dW_u\right)^+\right]-\sigmad E^P\left[ \left(\int_0^s \overline{\sigma}dW_u\right)^-\right]\right)ds\\
&\;=\int_0^t\Bigg(\sigmau\left(\overline{\sigma}\int_0^\infty x\frac{1}{\sqrt{2\pi}}\sqrt{s}e^{-\frac{x^2}{2}}dx\right)+\\
&\qquad\qquad\qquad\qquad+\sigmad \left(\overline{\sigma}\int_{-\infty}^0 x\frac{1}{\sqrt{2\pi}}\sqrt{s}e^{-\frac{x^2}{2}}dx\right)\Bigg)ds\\
&\;=\frac{\sigmau-\sigmad}{\sqrt{2\pi}}\overline{\sigma}\int_0^t\sqrt{s}ds=\frac{\sigmau-\sigmad}{\sqrt{2\pi}}\overline{\sigma}\frac{2}{3}t^{3/2}.
\end{split}
\]
Finally note that in the same way we can compute
\[
\begin{split}
\esp{-B_t\langle B \rangle_t}&=\esp{-\int_0^tB_sd\langle B\rangle_s}\leq\esp{-\int_0^tG(-B_s)ds}\\
&=E^{P^{\overline{\sigma}}}\left[\int_0^t2G(-B_s)ds\right]=\frac{\sigmau-\sigmad}{\sqrt{2\pi}}\overline{\sigma}\frac{2}{3} t^{3/2},
\end{split}
\]
which is precisely the same as $\esp{\int_0^t2G(B_s)ds}$.
\end{proof}

\section{Robust Mean-Variance Hedging}
\subsection{The Setting}
We start by fixing a finite time horizon $T$ and the measurable space $(\Omega,\mathcal{F})$, where $\Omega:= \{\omega \in C([0,T],\mathbb{R}) : \omega(0) = 0\}$, $\mathbb{F}:= \{\mathcal{F}_t, t\in[0,T]\}$ is the filtration generated by the canonical process $B$ and $\mathcal{F}=\mathcal{F}_T$.
\begin{oss}
This choice of the measurable space will allow us to use the results on stochastic calculus with respect to the $G$-Brownian motion and in particular the $G$-martingale representation Theorem~\ref{teo:decomposition} presented in Section 2. This assumption can be done without loss of generality as, for any probability measure $P$ on $(\Omega,\mathcal{F})$ denoting with $\bar{\mathbb{F}}^P:= \{\bar{\mathcal{F}}^P_t, t\in[0,T]\}$ the $P$-augmented filtration, we have the following lemma (see~\cite{soner:representation} for the proof).
\begin{lemma}  
For any $\bar{\mathcal{F}}_t^P$-measurable random variable $\xi$, there exists a unique (P-a.s.) $\mathcal{F}_t$-measurable random variable $\tilde{\xi}$ such that $\tilde{\xi}=\xi$, P-a.s..
Similarly, for every $\bar{\mathcal{F}}_t^P$-progressively measurable process $X$, there exists a unique $\mathcal{F}_t$-progressively measurable process $\tilde{X}$ such that $\tilde{X}=X$, $dt \otimes dP$-a.e.. Moreover, if $X$ is $P$-almost surely continuous, then one can choose $\tilde{X}$ to be $P$-almost surely continuous.
\end{lemma}
\end{oss}
\noindent
We consider the following discounted assets
\[
\begin{cases}
dX_t=X_tdB_t,  & X_0>0,\\
d\gamma_t=0, & \gamma_0=1,
\end{cases}
\] 
where $(\gamma_t)\ut$ denotes the discounted risk-free asset. In analogy to what is done in~\cite{Schweizer}, we take into consideration the space of strategies of the following type.
\begin{defn}
\label{defn:strategia}
A trading strategy $\varphi=(\phi_t,\eta_t)\ut$ is called admissible if $(\phi_t)\ut\in\Phi$, where
\[
\Phi:=\left\{\phi\text{ predictable }\Big| \esp{\left(\int_0^T\phi_tX_tdB_t\right)^2}<\infty\right\},
\]
$\eta$ is adapted, and it is self-financing, i.e.
\[
V_t(\varphi)=\eta_t\gamma_t+\phi_t X_t =V_0+\int_0^t\phi_sdX_s, \qquad \forall\;t\in[0,T].
\]
\end{defn} 

\noindent
The value of such strategies $\varphi\in\Phi$ at any time $t\in[0,T]$ is then completely determined by $(V_0,\phi)$, so that we can write $V_t(\varphi)=V_t(V_0,\phi)$ for all $t\in[0,T]$. 

We consider the problem of hedging a contingent claim $H\in L_G^{2+\epsilon}(\F_T)$, for an $\epsilon>0$, using admissible trading strategies. This integrability condition on $H$ is required in order to be able to use the $G$-martingale representation theorem. As a claim $H$ can be perfectly replicated with such a strategy only if it is symmetric, for a general derivative $H$ the idea of robust mean-variance hedging is to minimize the residual terminal risk defined as
\begin{equation}
\label{eq:mean-variance}
J_0(V_0,\phi):=E_G\left[\left(H-V_T(V_0,\phi)\right)^2\right]=\sup_{P\in\Pro}E^P\left[\left(H-V_T(V_0,\phi)\right)^2\right]
\end{equation}
by the choice of $(V_0,\phi)$. That is we wish to solve 
\begin{equation}
\label{problemadef}
\inf_{(V_0,\phi)\in\mathbb{R}_+\times \Phi } J_0(V_0,\phi)=\inf_{(V_0,\phi)\in\mathbb{R}_+\times \Phi } \esp{\left(H-V_T(V_0,\phi)\right)^2},
\end{equation}
as it is done in~\cite{Schweizer} in the classical case in which a unique prior exists. If an optimal $(V_0^\ast,\phi^\ast)\in\mathbb{R}_+\times \Phi$ exists for the problem~\eqref{problemadef}, we call $\phi^\ast$ optimal mean-variance strategy with optimal mean-variance portfolio 
\[
V_t=V_0^\ast +\int_0^t\phi_s^\ast dX_s, \quad t\in[0,T].
\]
The functional in~\eqref{problemadef} can be interpreted as a stochastic game between the agent and the market, the latter displaying the worst case volatility scenario and the former choosing the best possible strategy. When we have $\Pro=\{P\}$ this problem is solved thanks to the Galtchouk-Kunita-Watanabe decomposition, by projecting $H$ onto the linear space $\{I=x+\int_0^T\phi_sdX_s\;| \;x\in \mathbb{R}\text{ and } \phi\in\Phi\}$ (for more on this in the classical case we refer again to~\cite{Schweizer}). Here the situation is more cumbersome for several reasons. Firstly, there exists no orthogonal decomposition of the space of $L_G^{2+\epsilon}$-integrable $G$-martingales. Moreover a symmetric criterion does not distinguish between a buyer or a seller, so the best hedging strategy should be optimal both for $H$ and $-H$. This prevents us from using straightforwardly the $G$-martingale representation theorem as the coefficients in the decomposition of $H$ are a priori different from those coming from the decomposition of $-H$, see Lemma~\ref{minus}. Nevertheless we can get some insights from its direct application. 
\begin{lemma}\label{boundv0}
The initial wealth $V_0^\ast$ of the optimal mean-variance portfolio lies in the interval $[-E_G[-H],E_G[H]]$.
\end{lemma}
\begin{proof}
Let
\begin{align}
\label{lab}H&=E_G[H]+\int_0^T\theta_sdB_s-K_T,\\
\notag -H&=E_G[-H]+\int_0^T\bar{\theta}_sdB_s-\bar{K}_T,
\end{align}
be the $G$-martingale decomposition of $H$ and $-H$ for suitable processes $(\theta_t)\ut$, $(\bar{\theta}_t)\ut$, $(K_t)\ut$ and $(\bar{K}_t)\ut$, as given in Theorem~\ref{teo:decomposition}, respectively. It then follows that
\begin{equation}
\label{mv}
\begin{split}
&E_G\left[\left(H-V_0-\int_0^T\phi_sX_sdB_s\right)^2\right]\\
=&E_G\left[\left(E_G[H]-V_0+\int_0^T(\theta_s-\phi_sX_s)dB_s-K_T\right)^2\right]\\
=&(E_G[H]-V_0)^2+E_G\Bigg[\left(\int_0^T\left(\theta_s-\phi_sX_s\right)dB_s-K_T\right)^2+\\
&\qquad\qquad\qquad\qquad-2K_T(E_G[H]-V_0)\Bigg],
\end{split}
\end{equation}
and similarly
\[
\begin{split}
&E_G\left[\left(-H+V_0+\int_0^T\phi_sX_sdB_s\right)^2\right]\\
=&E_G\left[\left(E_G[-H]+V_0+\int_0^T(\bar{\theta}_s+\phi_sX_s)dB_s-\bar{K}_T\right)^2\right]\\
=&(E_G[-H]+V_0)^2+E_G\Bigg[\left(\int_0^T\left(\bar{\theta}_s+\phi_sX_s\right)dB_s-\bar{K}_T\right)^2+\\
&\qquad\qquad\qquad\qquad\qquad-2\bar{K}_T(E_G[-H]+V_0)\Bigg],
\end{split}
\]
by the properties of the stochastic integrals with respect to the $G$-Brownian motion and Proposition~\ref{prop:useful}.
From the expressions above we see that, as $K_T$ and $\bar{K}_T$ are strictly positive random variables, the optimal initial wealth $V_0^\ast$ is in the interval $[-E_G[-H],E_G[H]]$.
\end{proof}
\noindent
This agrees with the results on no-arbitrage pricing presented in~\cite{Vorbrink}, thanks to which we can argue that $V_0$ should indeed be in $(-E_G[-H],E_G[H])$, as long as $-E_G[-H]<E_G[H]$. When the claim is symmetric, i.e.\ $\esp{H}=-\esp{-H}$, it is also perfectly replicable and we would then have $V_0^\ast=\esp{H}$ and $(\phi^\ast_t)\ut=(\theta_t/X_t)\ut$, as in the classical case.

As for the initial value, it is possible to show that also the optimal trading strategy must belong to some bounded set in the $M^2_G$ norm.

\begin{lemma}\label{boundphi}
Let be given a contingent claim $H\in L_G^{2+\epsilon}(\mathcal{F}_T)$ with
\[
H=E_G[H]+\int_0^T\theta_sdB_s-K_T,
\]
for some $(\theta_t)\ut\in M_G^2(0,T)$ and $K_T\in L_G^2(\mathcal{F}_T)$. Then there exists a $R\in\mathbb{R}_+$ such that 
\[
\inf_{(V_0,\phi)\in\mathbb{R}_+\times \Phi } J_0(V_0,\phi)=\inf_{\overset{(V_0,\phi)\in\mathbb{R}_+\times \Phi}
{\Vert \int_0^T(\theta_s-\phi_sX_s)dB_s  \Vert_2}\leq R} J_0(V_0,\phi).
\]
\end{lemma} 

\begin{proof}
We start by noticing that the optimal mean variance portfolio $(V_0^\ast, \phi^\ast)$ clearly satisfies 
\begin{equation}\label{M1}
J(V_0^\ast, \phi^\ast)\leq \esp{H^2}
\end{equation}
and put 
\begin{align*}
A&:= \esp{H}-V_0-K_T,\\
D&:= \int_0^T (\theta_s-\phi_s X_s) dB_s.
\end{align*}
We can derive the following chain of inequalities
\begin{align*}
J(V_0,\phi)&=\esp{\left( A+D\right)^2}=\esp{A^2+D^2+2AD}\\
&\geq \esp{D^2}-\esp{-D^2}-\esp{-2AD}\\
&\geq \esp{D^2}-\esp{-A^2}-2\esp{A^2}^\frac{1}{2}\esp{D^2}^\frac{1}{2}.
\end{align*}
This shows that for great values of $\esp{D^2}$, i.e.\ when the $L^2_G$ distance of $\int_0^T \phi_s X_s dB_s$ from $\int_0^T \theta_s dB_s$ is too big, for any $V_0\in(-\esp{-H},\esp{H})$ the terminal risk $J(V_0,\phi)$ cannot be smaller than the upper bound in \eqref{M1}. This completes the proof.
\end{proof}

\begin{teo}
\label{theconv}
Let be given a claim $H\in L_G^{2+\epsilon}(\mathcal{F}_T)$ and a sequence of random variables $(H^n)_{n\in\mathbb{N}}$ such that $\Vert H - H^n\Vert_{2+\epsilon}\to 0$ as $n\to\infty$. Then as $n\to\infty$ we have
\begin{equation*}
J_n^\ast\to J^\ast,
\end{equation*}
where, for every $n\in\mathbb{N}$,
\[
J_n^\ast:=\inf_{(V_0,\phi)\in\mathbb{R}_+\times \Phi } \esp{\left(H^n-V_T(V_0,\phi)\right)^2}
\]
and 
\[
J^\ast:=\inf_{(V_0,\phi)\in\mathbb{R}_+\times \Phi } \esp{\left(H-V_T(V_0,\phi)\right)^2}.
\]
\end{teo}

\begin{proof}
As first step of the proof we study the convergence of  the terminal risk 
\begin{equation}	\label{majo6}
\esp{\left(H^n-V_T(V_0,\phi)\right)^2}\to\esp{\left(H-V_T(V_0,\phi)\right)^2},
\end{equation}
for some strategy $(V_0,\phi)$. We assume without loss of generality that $H$ has a representation as in \eqref{lab}. Similarly, for every $n\in\mathbb{N}$, we claim that
\[
H^n=\esp{H^n}+\int_0^T \theta_s^ndB_s-K^n_T,
\]
for a $(\theta^n_t)\ut\in M_G^2(0,T)$ and $K_T^n\in L_G^2(\mathcal{F}_T)$. We begin by proving that we can restrict ourselves to study the convergence in \eqref{majo6} for a \emph{bounded class of trading strategies}. It follows from Theorem 4.5 in \cite{song} that the $L^2_G$ convergence of $(H^n)_{n\in\mathbb{N}}$ to $H$ implies also 
\[
\Vert \int_0^T\left(\theta_s^n-\theta_s\right)dB_s \Vert_2\to 0
\]
and $\Vert K_T^n - K_T\Vert_2\to 0$ as $n\to\infty$. These facts, together with Lemma \ref{boundv0} and Lemma \ref{boundphi}, allow us to fix a $R\in\mathbb{R}_+$ such that 
\begin{align*}
J_n^\ast&=\inf_{\overset{(V_0,\phi)\in\mathbb{R}_+\times \Phi}
{\Vert V_0 +\int_0^T(\theta_s-\phi_sX_s)dB_s  \Vert_2}\leq R} \esp{\left(H^n-V_T(V_0,\phi)\right)^2}\\ 
J^\ast&=\inf_{\overset{(V_0,\phi)\in\mathbb{R}_+\times \Phi}
{\Vert V_0+\int_0^T(\theta_s-\phi_sX_s)dB_s  \Vert_2}\leq R} \esp{\left(H-V_T(V_0,\phi)\right)^2}.
\end{align*}
This in turns implies the convergence 
\[
\esp{\left(H^n-\cdot\right)^2}\to\esp{\left(H-\cdot\right)^2}
\]
on the set of strategies $(V_0,\phi)\in\mathbb{R}_+\times \Phi$ such that $\Vert V_0+\int_0^T(\theta_s-\phi_sX_s)dB_s  \Vert_2\leq R$. In fact, denoting $x:= V_0 +\int_0^T\phi_s X_s dB_s$ any of such strategies, for any $\delta>0$ we can find $\bar{n}\in\mathbb{N}$ such that for all $n>\bar{n}$
\begin{align}
\notag&\left|\esp{\left(H^n-x\right)^2}-\esp{\left(H-x\right)^2}\right|\leq \left|\esp{\left(H^n-x\right)^2-\left(H-x\right)^2}\right|\\
\notag\leq&\esp{|\left(H^n-x\right)^2-\left(H-x\right)^2|}=\esp{|\left(H^n-H\right)\left(H^n+H-2x\right)|}\\
\notag\leq&\esp{\left(H^n-H\right)^2}^\frac{1}{2}\esp{\left(H^n+H-2x\right)^2}^\frac{1}{2}\\
\label{majo}\leq&\esp{\left(H^n-H\right)^2}^\frac{1}{2}\left(\esp{\left(H^n+H\right)^2}^\frac{1}{2}+\esp{(2x)^2}^\frac{1}{2}\right)<\delta.
\end{align}
This is clear since the second factor in \eqref{majo} is bounded. The previous chain of inequalities holds true also upon considering the supremum of $x$ over the set $\Vert x \Vert_2\leq R$, which in turns implies uniform convergence. We can now prove the main statement. For any $\delta>0$, from the definition of $J^\ast$, there exists $(\bar{V}_0,\bar{\phi})\in\mathbb{R}_+\times \Phi$ such that $\Vert \bar{V}_0+\int_0^T(\theta_s-\bar{\phi}_sX_s)dB_s  \Vert_2\leq R$ and 
\begin{equation}\label{majo1}
J^\ast+\delta \geq \esp{\left(H-\bar{V}_0-\int_0^T\bar{\phi}_sX_sdB_s\right)^2}.
\end{equation}
Moreover, the uniform convergence from \eqref{majo}, allows us to consider $n$ big enough so that 
\begin{equation}\label{majo2}
\left|\esp{\left(H-\bar{V}_0-\int_0^T\bar{\phi}_sX_sdB_s\right)^2}- \esp{\left(H^n-\bar{V}_0-\int_0^T\bar{\phi}_sX_sdB_s\right)^2}\right|<\delta.
\end{equation}
From \eqref{majo1} and \eqref{majo2} we can conclude that 
\begin{equation}\label{majo3}
J^\ast+2\delta\geq J_n^\ast.
\end{equation}
Analogously it is possible to find $(\tilde{V}_0,\tilde{\phi})$ such that 
\[
J_n^\ast+\delta\geq \esp{\left(H^n-\tilde{V}_0-\int_0^T\tilde{\phi}_sX_sdB_s\right)^2}
\]
and 
\[
\left|\esp{\left(H-\tilde{V}_0-\int_0^T\tilde{\phi}_sX_sdB_s\right)^2}- \esp{\left(H^n-\tilde{V}_0-\int_0^T\tilde{\phi}_sX_sdB_s\right)^2}\right|<\delta,
\]
from which we can argue
\begin{equation}\label{majo4}
J_n^\ast\geq J^\ast-2\delta.
\end{equation}
The inequalities \eqref{majo3} and \eqref{majo4} conclude the proof as together they imply
\[
J^\ast-2\delta\leq J_n^\ast\leq J^\ast+2\delta
\]
and $\delta$ was chosen arbitrarily.
\end{proof}

\begin{oss}
Theorem \ref{theconv} shows that we can begin our study of the mean-variance optimization by considering claims in the space $L_{ip}(\F_T)$. Any random variable in $L_G^{2+\epsilon}(\F_T)$ is in fact by definition the limit in the $L_G^{2+\epsilon}$-norm of elements in $L_{ip}(\F_T)$. Moreover, as stated in Theorem~\ref{decopeng}, this class of random variables has the great advantage that the term $-K_T$ in their representation has a further decomposition as 
\begin{equation}
\label{decompositionk}
-K_T=\int_0^T\eta_s d\brac_s-2\int_0^TG(\eta_s)ds,
\end{equation}
for some process $(\eta_t)\ut\in M^1_G(0,T)$. 
\end{oss}

\noindent
From now on we consider $H\in L_G^{2+\epsilon}(\F_T)$ with decomposition 
\begin{equation}
\label{star}
\begin{split}
H&=\esp{H}+\int_0^T\theta_s dB_s - K_T(\eta)\\
&=\esp{H}+\int_0^T\theta_s dB_s+\int_0^T\eta_s d\langle B\rangle_s-2\int_0^TG(\eta_s) ds.
\end{split}
\end{equation}
Given the complexity of the problem, we proceed stepwise as follows. We first enforce some conditions on the process $\eta$, namely being deterministic or maximally distributed, then we assume $\eta$ to be a piecewise constant process having some particular characteristics that we will clarify at each time. In these cases we are able to solve the mean-variance hedging problem explicitly. Finally we address the general case by providing estimates of the minimal terminal risk.

\section{Explicit Solutions}
We first present the computation of the optimal mean-variance portfolio for random variables $H\in L_G^{2+\epsilon}(\F_T)$ with decomposition~\eqref{star}, where $\eta$ is assumed to be deterministic or depending only on the realization of $(\langle B\rangle_t)\ut$. On the contrary the integrand $\theta$ in~\eqref{star} is completely general and must only belong to $M_G^2(0,T)$. In this way, as $\eta$ does not exhibit volatility uncertainty through a direct dependence on the $G$-Brownian motion, uncertainty can be hedged by means of the initial wealth $V_0$ without using the strategy $\phi$. In these cases we are able to provide explicitly the optimal solutions in Theorem~\ref{deterministiceta} and Theorem~\ref{particularcase}.
\subsection{Deterministic $\eta$}
We first consider the case where $\eta$ in the representation~\eqref{star} is deterministic, and provide the optimal investment strategy and initial wealth.
\begin{teo} 
\label{deterministiceta}
Consider a claim $H\in L_G^{2+\epsilon}(\F_T)$ of the following form
\begin{equation}
\label{H}
H=\esp{H}+\int_0^T\theta_sdB_s+\int_0^T\eta_sd\langle B\rangle_s-\int_0^T2G(\eta_s)ds,
\end{equation}
where $\theta\in M_G^{2}(0,T)$ and $\eta\in M^1_G(0,T)$ is a deterministic process. The optimal mean-variance portfolio is given by
\[
\phi^\ast_tX_t = \theta_t
\] 
for every $t\in[0,T]$ and
\[
V_0^\ast=\frac{\esp{H}-\esp{-H}}{2}.
\]
\end{teo}
\begin{proof}
We start by computing the span of the process
\[
\esp{H}+\int_0^t\eta_sd\langle B\rangle_s-\int_0^T2G(\eta_s)ds.
\]
This lies quasi surely in the interval $[\esp{H}-(\overline{\sigma}^2-\underline{\sigma}^2)\int_0^T|\eta_s|ds, \esp{H}]$. We begin with the upper bound, noticing that under the volatility scenario given by
\[
\tilde{\sigma}_t=
\begin{cases} 
\overline{\sigma}^2 & \text{if $\eta_t\geq 0$}, \\
\underline{\sigma}^2 & \text{if $\eta_t< 0$},
\end{cases}
\]
for each $t\in[0,T]$, the negative random variable $\int_0^T\eta_sd\langle B\rangle_s-\int_0^T2G(\eta_s)ds$ is $P^{\tilde{\sigma}}$-a.s.\ equal to zero. As a consequence we have that $E^{P^{\tilde{\sigma}}}[H]=\esp{H}$. For the lower bound we consider 
\[
\tilde{\sigma}'_t=
\begin{cases} 
\overline{\sigma}^2 & \text{if $\eta_t\leq 0$}, \\
\underline{\sigma}^2 & \text{if $\eta_t> 0$},
\end{cases}
\]
for each $t\in[0,T]$. This is the scenario where $\int_0^T\eta_sd\langle B\rangle_s-\int_0^T2G(\eta_s)ds$ reaches its minimum. It follows that $E^{P^{\tilde{\sigma}'}}[H]=-\esp{-H}$. In fact, from~\eqref{H},
\begin{align}
-H&\notag=-\esp{H}-\int_0^T\theta_sdB_s-\int_0^T\eta_sd\langle B\rangle_s+\int_0^T2G(\eta_s)ds\\
&\notag=-\esp{H}-\int_0^T\theta_sdB_s-\int_0^T\eta_sd\langle B\rangle_s+\int_0^T2G(\eta_s)ds\\
&\notag\qquad\qquad\qquad\qquad+\int_0^T2G(-\eta_s)ds-\int_0^T2G(-\eta_s)ds\\
&\label{-H1}=-\esp{H}+\int_0^T(-\theta_s)dB_s+\int_0^T(-\eta_s)d\langle B\rangle_s-\int_0^T2G(-\eta_s)ds\\
&\notag \qquad\qquad\qquad\qquad+(\overline{\sigma}^2-\underline{\sigma}^2)\int_0^T|\eta_s|ds,
\end{align}
since
\[
\int_0^T2(G(\eta_s)+G(-\eta_s))ds=(\overline{\sigma}^2-\underline{\sigma}^2)\int_0^T|\eta_s|ds.
\]
We note that the expression~\eqref{-H1}, as $\eta$ is deterministic, provides the $G$-martingale decomposition of $-H$. Hence we can conclude that  
\begin{equation}
\label{tag3}
-\esp{H}+(\overline{\sigma}^2-\underline{\sigma}^2)\int_0^T|\eta_s|ds=\esp{-H}.
\end{equation}
Then, using Proposition~\ref{jensen} together with Lemma~\ref{quadrato} we get
\begin{align}
&\notag\inf_{(V_0,\phi)} E_G\Big[(\esp{H}-V_0+\int_0^T(\theta_s-\phi_sX_s)dB_s+\int_0^T\eta_sd\langle B\rangle_s+\\
&\notag-\int_0^T2G(\eta_s)ds)^2\Big]
\end{align}
\begin{align}
&\notag\geq\inf_{(V_0,\phi)} \Biggl(E_G\Big[\esp{H}-V_0+\int_0^T(\theta_s-\phi_sX_s)dB_s+\int_0^T\eta_sd\langle B\rangle_s+\\
&\notag-\int_0^T2G(\eta_s)ds\Big]^2\vee E_G\Big[-\esp{H}+V_0-\int_0^T(\theta_s-\phi_sX_s)dB_s+\\
&\notag-\int_0^T\eta_sd\langle B\rangle_s+\int_0^T2G(\eta_s)ds\Big]^2\Biggl)\\
&\label{tag1}=\inf_{V_0} \Biggl(E_G\Big[\esp{H}-V_0+\int_0^T\eta_sd\langle B\rangle_s-\int_0^T2G(\eta_s)ds\Big]^2\;\;\vee \\
&\notag E_G\Big[-\esp{H}+V_0-\int_0^T\eta_sd\langle B\rangle_s+\int_0^T2G(\eta_s)ds\Big]^2\Biggl)\\
&\label{tag2}=\inf_{V_0} \Biggl(E_G\Big[\esp{H}-V_0+\int_0^T\eta_sd\langle B\rangle_s-\int_0^T2G(\eta_s)ds\Big]^2\;\;\vee \\
&\notag E_G\Big[\esp{-H}+V_0+\int_0^T(-\eta_s)d\langle B\rangle_s-\int_0^T2G(-\eta_s)ds\Big]^2\Biggl),
\end{align}
where we have used Proposition~\ref{prop:useful} in~\eqref{tag1} and the relation~\eqref{tag3} in~\eqref{tag2}. This is equal to
\begin{equation}
\label{v}
\inf_{V_0} \Biggl(E_G\Big[\esp{H}-V_0\Big]^2\;\;\vee\;\;E_G\Big[\esp{-H}+V_0\Big]^2\Biggl),
\end{equation}
as 
\begin{align*}
&E_G\Big[a+\int_0^T\xi_sd\langle B\rangle_s-\int_0^T2G(\xi_s)ds\Big]=\\
=& a+E_G\Big[\int_0^T\xi_sd\langle B\rangle_s-\int_0^T2G(\xi_s)ds\Big]=a,
\end{align*}
for $a\in\mathbb{R}$ and $\xi\in M^1_G(0,T)$. The minimum of~\eqref{v} is attained for $V_0^\ast=\frac{\esp{H}-\esp{-H}}{2}$ and is equal to $\left(\frac{\esp{H}+\esp{-H}}{2}\right)^2$. If we show that
\begin{align*}
&\esp{\left(\esp{H}-V_0^\ast+\int_0^T\eta_sd\langle B\rangle_s-\int_0^T2G(\eta_s)ds\right)^2}\\
&\qquad\qquad\qquad\qquad\qquad\qquad=\left(\frac{\esp{H}+\esp{-H}}{2}\right)^2
\end{align*}
the proof is completed. Since
\[
\esp{H}+\int_0^T\eta_sd\langle B\rangle_s-\int_0^T2G(\eta_s)ds
\]
lies between $\esp{H}-(\overline{\sigma}^2-\underline{\sigma}^2)\int_0^T|\eta_s|ds=-\esp{-H}$ and $\esp{H}$, it is clear that the maximum of 
\[
|\esp{H}-V_0^\ast+\int_0^T\eta_sd\langle B\rangle_s-\int_0^T2G(\eta_s)ds| 
\]
under the constraint $V_0^\ast\in[-\esp{-H},\esp{H}]$ is given by $\frac{\esp{H}+\esp{-H}}{2}$. This completes the proof.
\end{proof}

\begin{oss}
\label{determ}
Note that the optimal investment strategy $\phi^\ast=\frac{\theta}{X}$ is well defined as $X$, being a geometric $G$-Brownian motion, is q.s.\ strictly greater than $0$. Moreover notice that, as 
\[
\int_0^T\eta_sd\langle B\rangle_s-\int_0^T2G(\eta_s)ds=-K_T,
\]
it holds
\[
\esp{H}-V_0^\ast=\frac{\esp{K_T}}{2},
\]
since
\[
\esp{K_T}=\esp{\esp{H}+\int_0^T\theta_sdB_s-H}=\esp{H}+\esp{-H}.
\]
\end{oss}

\begin{oss}
Notice that in a context in which a unique prior exists, i.e.\ $\overline{\sigma}=\underline{\sigma}$, $E[H]=\esp{H}=-\esp{-H}$, the optimal initial wealth and strategy derived in Theorem~\ref{deterministiceta} are consistent with the results on mean-variance hedging in the classical framework.
\end{oss}

\noindent
The set of contingent claims which admit the decomposition~\eqref{H} for $\eta$ deterministic is non trivial. For any given integrable deterministic process $(\eta_t)\ut$, any constant $c\in\mathbb{R}$ and any process $(\theta_t)\ut\in M_G^2(0,T)$, we can construct the claim
\[
H:=c+\int_0^T\theta_s dB_s +\int_0^T\eta_sd\langle B\rangle_s-\int_0^t2G(\eta_s)ds,
\]
for which the result of Theorem~\ref{deterministiceta} holds. The intersection of such a set of random variables with $L_{ip}(\F_T)$ includes the second degree polynomials in $(B_{t_0},B_{t_1}-B_{t_0},\dots,B_{t_n}-B_{t_{n-1}})$, where $\{t_i\}_{i=0}^n$ is a partition of $[0,T]$. To have an intuition on this fact consider for simplicity random variables depending only on one increment of the $G$-Brownian motion. The coefficients of the decomposition of $H=\varphi(B_T-B_0)$ are given by 
\[
\eta_t(\omega)=\partial_x^2u(t,\omega)
\] 
and 
\[
\theta_t(\omega)=\partial_xu(t,\omega),
\]
where $u$ is the solution to 
\[
\begin{cases}
\partial_t u + G(\partial^2_x u)=0,\\
u(T,x)=\varphi(x),
\end{cases}
\]
for $(t,x)\in[0,T]\times \mathbb{R}$ (see~\cite{Peng:dynamicrisk}). If $\eta$ is deterministic, we can write $\partial_x^2u(t,\omega)$ as a function of $t$, i.e.\ $a(t):=\partial_x^2u(t,\omega)$. Therefore, by integration w.r.t. $x$, we see that $u(t,x)$ must be of the form 
\[
u(t,x)=\frac{a(t)}{2}x^2+b(t)x+c(t),
\]
so that 
\[
H=\frac{a(T)}{2}B_T^2+b(T)B_T+c(T).
\]

\begin{oss}\label{detrmrem}
Another class of claims that can be optimally hedged by means of Theorem \ref{deterministiceta} is obtained thanks to Theorem 4.1 in \cite{Vorbrink}. If we consider the situation in which $H=\Phi(X_T)$, for some real valued Lipschitz function $\Phi$, then it holds (see \cite{Vorbrink} for the details)
\begin{align*}
\Phi(X_T)&= \esp{\Phi(X_T)}+\int_0^T \partial_xu(t,X_t)X_tdB_t \\
& \qquad\qquad+\frac{1}{2}\int_0^T \partial_{xx} u(t,X_t)X^2_td\langle B\rangle_t - \int_0^T G(\partial_{xx} u(t,X_t))X^2_tdt,
\end{align*}
where $u$ solves 
\[
\begin{cases}
\partial_t u + G(x^2\partial^2_x u)=0,\\
u(T,x)=\Phi(x).
\end{cases}
\]

\noindent
It is then easy to see that $\partial_{xx} u(t,X_t)X^2_t$ is deterministic for every $t\in[0,T]$ if and only if 
\[
H=\Phi(X_T)=u(T,X_T)= a(T)\log X_T+b(T) X_T+ c(T),
\]
for some real functions $a,b$ and $c$.\\
Through a slight modification to the previous argument we can prove that if on the market there exists another asset $X^\prime$, which is not possible to trade and solves the SDE
\[
dX^\prime_t=\alpha(X^\prime_t)dB_t,\qquad X^\prime_0>0,
\]
for some Lipschitz function $\alpha$, then it is possible to use again Theorem \ref{deterministiceta} to hedge every claim $\Phi(X^\prime_T)$, where $\Phi$ is a Lipschitz function such that 
\[
\begin{cases}
\partial_t u + G(\alpha^2(x)\partial^2_x u)=0,\\
u(T,x)=\Phi(x),
\end{cases}
\] 
provided that $\partial_{xx}u(t,x)=\frac{1}{\alpha(x)}$ for every $(t,x)\in[0,T]\times \mathbb{R}$.
\end{oss}

\subsection{Maximally Distributed $\eta$}
We now consider the case in which $\eta$ only shows mean uncertainty, being a function of the quadratic variation of the $G$-Brownian motion. Also in this case we are able to retrieve a complete description of the optimal mean-variance portfolio.
\begin{teo}
\label{particularcase}
Let $H\in L_G^{2+\epsilon}(\F_T)$ be of the form 
\[
H=\esp{H}+\int_0^T\theta_s d B_s+\int_0^T \psi(\brac_s)d\brac_s-2\int_0^TG(\psi(\brac_s))ds,
\]
where $(\theta_t)\ut\in M^2_G(0,T)$ and $\psi: \R\to\R$ is such that there exist $k\in\R$ and $\alpha\in\R_+$ for which
\[
|\psi(x)-\psi(y)|\leq \alpha|x-y|^k,
\]
for all $x,y\in\R$. The optimal mean-variance portfolio is given by 
\[
\phi_t^\ast X_t=\theta_t
\]
for every $t\in[0,T]$ and
\[
V^\ast_0=\frac{\esp{H}-\esp{-H}}{2}.
\]
\end{teo}
\begin{proof}
As in Theorem~\ref{deterministiceta}, we start by applying the $G$-Jensen's inequality to obtain
\begin{align}
\notag&\;\esp{\left(c+\int_0^T\varphi_s d B_s+\int_0^T \psi(\brac_s)d\brac_s-2\int_0^TG(\psi(\brac_s))ds\right)^2}\\
\notag&\geq \esp{c+\int_0^T\varphi_s d B_s+\int_0^T \psi(\brac_s)d\brac_s-2\int_0^TG(\psi(\brac_s))ds}^2\vee\\
\notag&\quad \esp{-c-\int_0^T\varphi_s d B_s-\int_0^T \psi(\brac_s)d\brac_s+2\int_0^TG(\psi(\brac_s))ds}^2\\
\label{aggiunta}&=c^2\vee \left(\esp{K_T}-c\right)^2,
\end{align}
where we defined
\begin{equation}\label{notation}
\begin{split}
c:&= \esp{H}-V_0,\\
\varphi_t:&=\theta_t-\phi_tX_t,
\end{split}
\end{equation}
for all $t\in[0,T]$. The minimum of~\eqref{aggiunta} is attained when $c^\ast=\frac{\esp{K_T}}{2}$, and it is equal to $\left(\frac{\esp{K_T}}{2}\right)^2$. We conclude by showing that this value is attained by choosing $V^\ast_0=\frac{\esp{H}-\esp{-H}}{2}$ and $\phi_t^\ast X_t=\theta_t$. We then compute
\begin{equation}
\label{bracc}
\begin{split}
&E_G\Bigg[\Big(\frac{\esp{K_T}}{2}+\int_0^T \psi(\brac_s)d\brac_s-2\int_0^TG(\psi(\brac_s))ds\Big)^2\Bigg].
\end{split}
\end{equation}
In order to do so we use a discretization, noting that 

\begin{equation}
\label{conver}
\psi^n(\brac_t):=\sum_{i=0}^{n-1}\psi(\brac_{t_i})\mathbb{I}_{[t_i,t_{i+1})}(t)\stackrel{M^2_G(0,T)}{\longrightarrow}\psi(\brac_t)
\end{equation}

\noindent
where $t_i=\frac{T}{n}i$. In fact
\[
\begin{split}
&\int_0^T\esp{|\psi(\brac_t)-\psi^n(\brac_t)|^2}dt=\sum_{i=0}^{n-1}\int_{t_i}^{t_{i+1}}\esp{|\psi(\brac_t)-\psi^n(\brac_t)|^2}dt\\
&\leq \sum_{i=0}^{n-1}\int_{t_i}^{t_{i+1}}\esp{|\brac_t-\brac_{t_i}|^{2k}}dt=n\int_0^{t_1}\esp{\brac_t^{2k}}dt=
n\int_0^{t_1}t^{2k}dt\\
&=\frac{n}{2k}\left(\frac{T}{n}\right)^{2k+1}\stackrel{n\to\infty}{\longrightarrow}0,
\end{split}
\]
and similarly for the convergence of $G(\psi^n(\brac_t))$ to $G(\psi(\brac_t))$. The expression in~\eqref{bracc} is then the limit when $n$ tends to infinity of

\begin{align}
\notag &\esp{\Big(\frac{\esp{K_T}}{2}+\sum_{i=0}^{n-1} \psi(\brac_{t_i})\deltabr\ipl-2\sum_{i=0}^{n-1} G(\psi(\brac\ii))\Delta t_{i+1}\Big)^2}\\
\notag =&E_G\Bigg[\Big(\frac{\esp{K_T}}{2}+\sum_{i=0}^{n-1} \psi\left(\sum_{j=0}^{i}\deltabr_{t_j}\right)\deltabr\ipl+\\
\notag &\qquad\qquad\qquad\qquad\qquad-2\sum_{i=0}^{n-1} G\left(\psi\left(\sum_{j=0}^{i}\deltabr_{t_j}\right)\right)\Delta t_{i+1}\Big)^2\Bigg]\\
\notag =&E_G\Biggl[E_G\Bigg[\Big(\frac{\esp{K_T}}{2}+\sum_{i=0}^{n-1} \psi\left(\sum_{j=0}^{i}\deltabr_{t_j}\right)\deltabr\ipl+\\
\notag &\qquad\qquad\qquad\qquad\qquad-2\sum_{i=0}^{n-1} G\left(\psi\left(\sum_{j=0}^{i}\deltabr_{t_j}\right)\right)\Delta t_{i+1}\Big)^2\Bigg|\F_{t_{n-1}}\Bigg]\Biggl]\\
\notag =&E_G\Biggl[\sup_{\sigmad\leq v_n\leq\sigmau}\Big(\frac{\esp{K_T}}{2}+\sum_{i=0}^{n-2} \psi\left(\sum_{j=0}^{i}\deltabr_{t_j}\right)\deltabr\ipl+\\
\notag &\qquad+\psi\left(\sum_{j=0}^{n-1}\deltabr_{t_j}\right)v_n\Delta t_n-2\sum_{i=0}^{n-1} G\left(\psi\left(\sum_{j=0}^{i}\deltabr_{t_j}\right)\right)\Delta t_{i+1}\Big)^2\Biggl]
\end{align}
\begin{align}
\notag =&E_G\Biggl[E_G\Bigg[\sup_{\sigmad\leq v_n\leq\sigmau}\Big(\frac{\esp{K_T}}{2}+\sum_{i=0}^{n-2} \psi\left(\sum_{j=0}^{i}\deltabr_{t_j}\right)\deltabr\ipl+\\
\notag &+\psi\left(\sum_{j=0}^{n-1}\deltabr_{t_j}\right)v_n\Delta t_n-2\sum_{i=0}^{n-1} G\left(\psi\left(\sum_{j=0}^{i}\deltabr_{t_j}\right)\right)\Delta t_{i+1}\Big)^2\Bigg| \F_{t_{n-2}}\Bigg]\Biggl]\\
\label{maxx} =&E_G\Biggl[\sup_{\overset{\sigmad\leq v_n\leq\sigmau}{\sigmad\leq v_{n-1}\leq\sigmau}}\Big(\frac{\esp{K_T}}{2}+\sum_{i=0}^{n-3} \psi\left(\sum_{j=0}^{i}\deltabr_{t_j}\right)\deltabr\ipl+\\
\notag &+\psi\left(\sum_{j=0}^{n-2}\deltabr_{t_j}\right)v_{n-1}\Delta t_{n-1}+\psi\left(\sum_{j=0}^{n-2}\deltabr_{t_j}+v_{n-1}\Delta t_{n-1}\right)v_n\Delta t_n+\\
\notag &-2G\left(\psi\left(\sum_{j=0}^{n-2}\deltabr_{t_j}+v_{n-1}\Delta t_{n-1}\right)\right)\Delta t_{n}+\\
\notag &-2\sum_{i=0}^{n-2} G\left(\psi\left(\sum_{j=0}^{i}\deltabr_{t_j}\right)\right)\Delta t_{i+1}\Big)^2\Biggl],
\end{align}
where we have used that $\deltabr$ is maximally distributed. Proceeding by iteration,~\eqref{maxx} is equal to 

\begin{equation}
\label{sup}
\begin{split}
&\sup_{\stackrel{\sigmad\leq v_i\leq \sigmau}{i=1,\dots,n}} \Bigg(\frac{\esp{K_T}}{2}+\sum_{i=0}^{n-1} \psi\left(\sum_{j=0}^{i}v_j\Delta t_j\right)v_{i+1}\Delta t_{i+1}+\qquad\qquad\qquad\qquad\text{ }\\
&\qquad\qquad\qquad\qquad\qquad-2\sum_{i=0}^{n-1} G\left(\psi\left(\sum_{j=0}^{i}v_j\Delta t_j\right)\right)\Delta t_{i+1}\Bigg)^2.
\end{split}
\end{equation}

\noindent
The supremum~\eqref{sup}, being a quadratic function of $(v_i)_{i=1,\dots,n}$, is attained either when the term depending on $(v_i)_{i=1,\dots,n}$ is equal to its minimum, which is zero, or its maximum, which is equal to
\[
\esp{2\sum_{i=0}^{n-1} G(\psi(\brac\ii))\Delta t_{i+1}-\sum_{i=0}^{n-1} \psi(\brac_{t_i})\deltabr\ipl}.
\]
In both cases, as $n$ tends to infinity the value of~\eqref{sup} converges to $\left(\frac{\esp{K_T}}{2}\right)^2$ because of~\eqref{conver}. 
\end{proof}

\noindent
As the optimal mean variance portfolio $(V_0^\ast, \phi^\ast)$ for a claim $H$ provides, via $(-V_0^\ast, -\phi^\ast)$, the optimal solution for the hedging of $-H$, the investment strategy $(\phi^\ast_t)\ut$ would not always be equal to the process $(\theta_t)\ut$ coming from the $G$-martingale decomposition of $H$ as in Theorem~\ref{teo:decomposition}. The result of Theorem~\ref{particularcase} does not contradict this intuition.

\begin{oss}
Using Lemma~\ref{minus} it is not difficult to prove that for contingent claims of the type
\[
H=\esp{H}+\int_0^T\theta_sdB_s+\sum_{i=0}^{n-1}\left(\psi(\brac\ii)\Delta\brac\ipl-2G(\psi(\brac\ii))\Delta t_{i+1}\right),
\]
where $(\theta_t)\ut\in M_G^2(0,T)$ and $\psi$ is a real continuous function, the decomposition of $-H$ has the expression
\[
-H=\esp{-H}+\int_0^T(-\theta_s)dBs -\bar{K}_T,
\]
for a suitable random variable $\bar{K}_T\in L_G^1(\F_T)$.
\end{oss}

\noindent
It is possible to use the same argument of Remark \ref{detrmrem} to characterize the class of contingent claims whose representation \eqref{H} exhibits an $\eta$ given by a function with polynomial growth of $\langle B\rangle$. This set includes the family of Lipschitz function of $\langle B\rangle$. Theorem \ref{particularcase} can be used to hedge \emph{volatility swaps}, i.e.\ $H=\sqrt{\langle B\rangle_T}-K$ with $K\in\mathbb{R}_+$, and other volatility derivatives (we refer to \cite{lee} for more details on volatility derivatives). In fact, given a Lipschitz function $\Phi$, the claim $\Phi(\langle B\rangle_T)$ can be written as 
\[
\begin{split}
\Phi(\langle B\rangle_T)&=\esp{\Phi(\langle B\rangle_T)}+\int_0^T \partial_x u(s,\langle B\rangle_s)\langle B\rangle_s d\langle B\rangle_s\\
&\qquad\qquad\qquad-2\int_0^T G(\partial_x u(s,\langle B\rangle_s))\langle B\rangle_s ds,
\end{split}
\]
where $u(t,x)$ solves
\[
\begin{cases}
\partial_t u + 2G(x\partial_x u)=0,\\
u(T,x)=\Phi(x),
\end{cases}
\] 
as a consequence of the nonlinear Feynman-Kac formula for $G$-Brownian motion (see \cite{nonlinear}) and the $G$-It\^o formula (see \cite{Peng:Gexp}).

\subsection{Piecewise Constant $\eta$}
We now study the optimal mean-variance portfolio for a broader class of claims, incorporating mean and volatility uncertainty in the process $\eta$. We first consider 
$$
\eta_s=\sum_{i=0}^{n-1}\eta\ii\mathbb{I}_{(t_i,t_{i+1}]}(s),
$$ 
for $n\in\mathbb{N}$, where $\{t_i\}_{i=0}^{n}$ is a partition of $[0,T]$, i.e.\ $0=t_0\leq t_1\leq\dots\leq t_n=T$, and $\eta_{t_i}\in L_{ip}(\F_{t_i})$ for all $i\in\{0,\dots,n\}$. We will outline a recursive solution procedure, which we are able to solve for $n=2$. In the case of $n>2$ the proof of Theorem~\ref{final2} provides a recursive procedure, which can be used to find numerically the optimal solution (see \cite{jac}). Finally we provide bounds for the optimal terminal risk~\eqref{problemadef} in Section 5.\\
As a preliminary result we restrict ourselves to the study of claims which can be represented in the following way
\begin{equation}
\label{eq:step}
H=\esp{H}+\theta_{t_1}\Delta B_{t_{2}}+\eta_{t_{1}} \deltabr_{t_{2}}-2G(\eta_{t_{1}})\Delta t_{2},
\end{equation}
where $0\leq t_1<t_{2}\leq T$ , $\theta_{t_1}\in L_G^2(\F_{t_1})$, $\Delta B_{t_2}:=B_{t_2}-B_{t_1}$ and similarly for $\deltabr_{t_2}$ and $\Delta t_{2}$. We choose accordingly the class of investment strategies $\phi$ of the form
\[
\phi_t=\phi_{t_1}\mathbb{I}_{(t_1,t_{2}]},
\]
where $\phi_{t_1}\in L_G^2(\F_{t_1})$. If we denote  
\[
\begin{split}
c:=\;& \esp{H}-V_0,\\
\varphi_t:=\;&\theta_t-\phi_tX_t,
\end{split}
\]
the risk functional~\eqref{eq:mean-variance} becomes
\begin{equation}
\label{eq:riskrexpr}
\begin{split}
&\esp{\left(\esp{H}-V_0+(\theta_{t_1}-\phi_{t_1} X_{t_1})\Delta B_{t_{2}}+\eta_{t_{1}} \deltabr_{t_{2}}-2G(\eta_{t_{1}})\Delta t_{2}\right)^2}\\
=&\esp{\left(c+\varphi_{t_1}\Delta B_{t_2}+\eta_{t_{1}} \deltabr_{t_{2}}-2G(\eta_{t_{1}})\Delta t_{2}\right)^2}\\
=&E_G\Big[\left(c+\eta_{t_1}\deltabr_{t_2}-2G(\eta_{t_1})\Delta t_{2}\right)^2+\varphi_{t_1}^2\Delta B_{t_2}^2+\\
&\qquad\qquad\qquad\qquad\quad+2\varphi_{t_1}\Delta B_{t_2}\eta_{t_1}\deltabr_{t_2}\Big],
\end{split}
\end{equation}
where we used Proposition~\ref{prop:useful} in the last step.

\begin{teo}
\label{1step}
Consider a claim $H\in L_G^{2+\epsilon}(\F_T)$ with decomposition as in~\eqref{eq:step}. The optimal mean-variance portfolio is given by $(V_0^\ast,\phi^\ast)$, where
\[
\phi^\ast X=\theta
\] 
and $V_0^\ast$ solves 
\begin{equation}
\label{52}
\inf_{V_0} \esp{(\esp{H}-V_0)^2\vee (\esp{H}-V_0-(\sigmau-\sigmad)\Delta t_{2}|\eta_{t_1}|)^2}.
\end{equation}
\end{teo}
\begin{proof}
\noindent
We start by computing

\begin{equation}
\label{eq:firstst}
\begin{split}
&\esp{\left(c+\eta_{t_1}\deltabr_{t_2}-2G(\eta_{t_1})\Delta t_{2}\right)^2}\\
=&\esp{\condesp[\F_{t_1}]{\left(c+\eta_{t_1}\deltabr_{t_2}-2G(\eta_{t_1})\Delta t_{2}\right)^2}}\\
=&\esp{f(\eta_{t_1})},
\end{split}
\end{equation}
where
\[
f(x)=\esp{\left(c+x\deltabr_{t_2}-2G(x)\Delta t_{2}\right)^2}.
\]
Using the fact that $\brac$ is maximally distributed,
\[
\begin{split}
f(x)&=\sup_{\sigmad\leq v\leq \sigmau} \left(c+xv\Delta t_{2}-2G(x)\Delta t_{2}\right)^2\\
&=\left(c+\sigmau x\Delta t_{2}-2G(x)\Delta t_{2}\right)^2\vee \left(c+\sigmad x\Delta t_{2}-2G(x)\Delta t_{2}\right)^2\\
&=c^2 \vee \left(c-(\sigmau-\sigmad)\Delta t_{2}|x|\right)^2,
\end{split}
\]
so that~\eqref{eq:firstst} becomes equal to 
\[
\begin{split}
\esp{c^2 \vee \left(c-(\sigmau-\sigmad)\Delta t_{2}|\eta_{t_1}|\right)^2}.
\end{split}
\]

\noindent
This means that, in the time interval $[t_1,t_{2}]$, the worst case scenario sets the volatility constantly equal to $\sigmau\Delta t_{2}$ when 
\[
c^2\geq \left(c-(\sigmau-\sigmad)\Delta t_{2}|\eta_{t_1}|\right)^2,
\]
which is equivalent to 
\[
c\geq \frac{(\sigmau-\sigmad)\Delta t_{2}|\eta_{t_1}|}{2},
\]
or to $\sigmad\Delta t_{2}$ if
\[
c\leq \frac{(\sigmau-\sigmad)\Delta t_{2}|\eta_{t_1}|}{2}.
\]
Hence it follows that, by Proposition~\ref{jensen}, for every $c\in(0,\esp{H}+\esp{-H})$
\begin{equation}
\label{eq:sp}
\begin{split}
&\inf_{\varphi}\esp{\left(c+\varphi_{t_1}\Delta B_{t_{2}}+\eta_{t_{1}} \deltabr_{t_{2}}-2G(\eta_{t_{1}})\Delta t_{2}\right)^2}\\
=&\inf_{\varphi}\esp{\condesp[\F_{t_1}]{\left(c+\varphi_{t_1}\Delta B_{t_{2}}+\eta_{t_{1}} \deltabr_{t_{2}}-2G(\eta_{t_{1}})\Delta t_{2}\right)^2}}\\
\geq &\inf_{\varphi} E_G\Big[\condesp[\F_{t_1}]{c+\varphi_{t_1}\Delta B_{t_{2}}+\eta_{t_{1}} \deltabr_{t_{2}}-2G(\eta_{t_{1}})\Delta t_{2}}^2\vee\\
&\quad\condesp[\F_{t_1}]{-c-\varphi_{t_1}\Delta B_{t_{2}}-\eta_{t_{1}} \deltabr_{t_{2}}+2G(\eta_{t_{1}})\Delta t_{2}}^2\Big]\\
&= \esp{c^2 \vee \left(c-(\sigmau-\sigmad)\Delta t_{2}|\eta_{t_1}|\right)^2}.
\end{split}
\end{equation} 
This allows us to conclude, as the lower bound is attained by choosing $\varphi_{t_1}=0$ and $V_0^\ast$ is the solution of~\eqref{52}.
\end{proof}

\noindent
Theorem~\ref{1step} shows that the determination of the optimal initial wealth can be more involved. We now show with a counterexample that the link between $\esp{K_T}$ and $V_0^\ast$ stated in Remark~\ref{determ} does not hold for general $\eta$. 
\begin{prop}
Let $H$ be of the form
\[
H=\esp{H}+\theta_{t_1}\Delta B_{t_{2}}+\eta_{t_{1}} \deltabr_{t_{2}}-2G(\eta_{t_{1}})\Delta t_{2},
\]
where $\theta_{t_1}\in L_G^2(\F_{t_1})$ and $\eta_{t_1}=e^{B_{t_1}}$. The optimal initial wealth of the mean-variance portfolio is different from 
\[
V_0^\ast=\frac{\esp{H}-\esp{-H}}{2}.
\]
\end{prop}

\begin{proof}
Let us first compute $\frac{\esp{H}+\esp{-H}}{2}$. By conditioning and using some results on the expectation of convex functions of the increments of the $G$-Brownian motion (see Proposition 11 in~\cite{Peng:Gexp}), we obtain
\[
\begin{split}
\esp{H}+\esp{-H}&=\esp{2G(e^{B_{t_1}})\Delta t_{2}-e^{B_{t_1}} \deltabr_{t_{2}}}\\
&=\esp{(\sigmau-\sigmad)\Delta t_{2}e^{B_{t_1}} }\\
&=E_P[(\sigmau-\sigmad)\Delta t_{2}e^{W_{t_1}\overline{\sigma}}]\\
&=(\sigmau-\sigmad)\Delta t_{2}e^{\sigmau t_1/2},
\end{split}
\]
where $(W_t)\ut$ is a standard Brownian motion under $P$. We now focus on the minimization over $c$ of 
\[
\begin{split}
H(c):=&\esp{c^2 \vee \left(c-(\sigmau-\sigmad)\Delta t_{2}e^{B_{t_1}}\right)^2}\\
=&E_P\left[c^2 \vee \left(c-(\sigmau-\sigmad)\Delta t_{2}e^{W_{t_1}\overline{\sigma}}\right)^2\right]\\
=&E_P\left[\left(\left(e^{W_{t_1}\overline{\sigma}}\Delta t_{2}(\sigmau-\sigmad)-c\right)^2-c^2\right)^+\right]+c^2\\
=&c^2+E_P\left[e^{W_{t_1}\overline{\sigma}}\Delta t_{2}(\sigmau-\sigmad)\left(e^{W_{t_1}\overline{\sigma}}\Delta t_{2}(\sigmau-\sigmad)-2c\right)^+\right]\\
=&c^2+E_P\left[e^{N\sqrt{t_1}\overline{\sigma}}\Delta t_{2}(\sigmau-\sigmad)\left(e^{N\sqrt{t_1}\overline{\sigma}}\Delta t_{2}(\sigmau-\sigmad)-2c\right)^+\right],
\end{split}
\]
where $N\sim \mathcal{N}(0,1)$ and we have used that 
\[
c^2 \vee \left(e^{B_{t_1}}\Delta t_{2}(\sigmau-\sigmad)-c\right)^2
\]
is a convex function of $B_{t_1}$. Let $y:=(\sigmau-\sigmad)\Delta t_{2}$ and 
\[
\begin{split}
A(x):&=\{x\in\R:e^{\overline{\sigma}\sqrt{t_{1}}x}>\frac{2c}{y}\}\\
&=\left\{x\in\R: x>\frac{\log\left(\frac{2c}{y}\right)}{\overline{\sigma}\sqrt{t_{1}}}\right\}\\
&=\left\{x\in\R: x>g(c)\right\},
\end{split}
\]
where $g(c):=\frac{\log\left(\frac{2c}{y}\right)}{\overline{\sigma}\sqrt{t_{1}}}$. With these notations $H(c)$ can be written as
\begin{align*}
H(c)=&c^2+E_P\left[e^{2\overline{\sigma}N\sqrt{t_{1}}}y^2\mathbb{I}_{A(N)}\right]-2cyE_P\left[e^{\overline{\sigma}
\sqrt{t_{1}}N}\mathbb{I}_{A(N)}\right]\\
=&c^2+y^2\int_{x>g(c)}e^{2\overline{\sigma}\sqrt{t_{1}}x}\frac{1}{\sqrt{2\pi}}e^{-\frac{x^2}{2}}dx-2cy\int_{x>g(c)}e^{\overline{\sigma}\sqrt{t_{1}}x}\frac{1}{\sqrt{2\pi}}e^{-\frac{x^2}{2}}dx.
\end{align*}

\noindent
We differentiate with respect to $c$ to find the stationary points:
\begin{align}
\notag H^\prime(c)&=2c-y^2e^{2\overline{\sigma}\sqrt{t_{1}}g(c)}\frac{1}{\sqrt{2\pi}}e^{-\frac{g(c)^2}{2}}g'(c)+2cye^{\overline{\sigma}\sqrt{t_{1}}g(c)}\frac{1}{\sqrt{2\pi}}e^{-\frac{g(c)^2}{2}}g'(c)+\\
\label{53} &-2y\int_{x>g(c)}\frac{1}{\sqrt{2\pi}}e^{\overline{\sigma}\sqrt{t_{1}}x-\frac{x^2}{2}}dx.
\end{align}

\noindent
We now substitute $c^\ast=\frac{\esp{H}+\esp{-H}}{2}=\frac{ye^{\frac{\sigmau t_1}{2}}}{2}$ into~\eqref{53} to see if it is a possible point of minimum. We obtain
\[
g(c^\ast)=\frac{\log\left(\frac{ye^{\frac{\sigmau t_1}{2}}}{y}\right)}{\overline{\sigma}\sqrt{t_{1}}}=\frac{1}{2}\overline{\sigma}\sqrt{t_{1}},
\]
and therefore
\begin{align*}
H^\prime \left(\frac{ye^{\frac{\sigmau t_1}{2}}}{2}\right)&=\;ye^{\frac{1}{2}\sigmau t_1}-y^2e^{2\overline{\sigma} \sqrt{t_1}\frac{1}{2}\overline{\sigma}\sqrt{t_1}}\frac{1}{\sqrt{2\pi}}e^{-\frac{g(c^\ast)^2}{2}}g^\prime(c^\ast)+\\
&\quad+2\frac{ye^{\frac{1}{2}\sigmau t_1}}{2}ye^{\overline{\sigma}\sqrt{t_1}\frac{1}{2}\overline{\sigma} \sqrt{t_1}}\frac{1}{\sqrt{2\pi}}e^{-\frac{g(c^\ast)^2}{2}}g^\prime(c^\ast)+\\
&\quad-2y\int_{x>\frac{1}{2}\overline{\sigma}\sqrt{t_{1}}}\frac{1}{\sqrt{2\pi}}e^{-\frac{1}{2}(x^2-2\overline{\sigma}\sqrt{t_{1}}x)}dx\\
&=y\left(e^{\frac{1}{2}\sigmau t_1}-2\int_{x>\frac{1}{2}\overline{\sigma}\sqrt{t_{1}}}\frac{1}{\sqrt{2\pi}}e^{-\frac{1}{2}(x^2-2\overline{\sigma}\sqrt{t_{1}}x)}dx\right)\\
&=y\left(e^{\frac{1}{2}\sigmau t_1}-2e^{\frac{1}{2}\sigmau t_1}\int_{z>-\frac{1}{2}\overline{\sigma}\sqrt{t_1}}\frac{1}{\sqrt{2\pi}}e^{-\frac{z^2}{2}}dz\right)\\
&=y\Bigg(e^{\frac{1}{2}\sigmau t_1}-2e^{\frac{1}{2}\sigmau t_1}\int_0^\infty\frac{1}{\sqrt{2\pi}}e^{-\frac{z^2}{2}}dz+\\
&\qquad\qquad-2e^{\frac{1}{2}\sigmau t_1}\int_{-\frac{1}{2}\overline{\sigma}\sqrt{t_1}}^0\frac{1}{\sqrt{2\pi}}e^{-\frac{z^2}{2}}dz\Bigg)\\
&=-2ye^{\frac{1}{2}\sigmau t_1}\int_{-\frac{1}{2}\overline{\sigma}\sqrt{t_1}}^0\frac{1}{\sqrt{2\pi}}e^{-\frac{z^2}{2}}dz,
\end{align*}
which is different from zero.
\end{proof}

\noindent
We now derive the optimal initial wealth for other particular cases, as we do in the following proposition. This result will constitute the first step of our recursive scheme. We remark that $\eta$ will now exhibit volatility uncertainty, which was excluded from the results in Sections 4.1 and 4.2, while the process $\theta\in M_G^2(0,T)$ is completely general.
\begin{prop}
\label{final1}
Consider a claim $H$ of the form 
\[
H=\esp{H}+\int_0^{t_2}\theta_{s} dB_{s}+\eta_{t_1}\deltabr_{t_2}-2G(\eta_{t_1})\Delta t_{2},
\]
where $0= t_{0}<t_1<t_{2}=T$, $(\theta_{s})_{s\in[0,t_2]}\in M_G^2(0,t_2)$, $\eta_{t_1}\in L_G^2(\F_{t_1})$ and 
\begin{equation}
\label{absvalue1}
|\eta_{t_1}|=\esp{|\eta_{t_1}|}+\int_0^{t_1}\mu_{s} dB_{s},
\end{equation}
for a certain process $(\mu_{s})_{s\in[0,t_1]}\in M_G^2(0,t_1) $. The optimal mean-variance portfolio is given by 
\[
X_t\phi_t^\ast=\left(\theta_{t}-\frac{\mu_{t}(\sigmau-\sigmad)\Delta t_2}{2}\right)\mathbb{I}_{(t_0,t_1]}(t)+\theta_{t}\mathbb{I}_{(t_1,t_2]}(t)
\]
for $t\in[0,T]$ and
\[
V_0^\ast=\frac{\esp{H}-\esp{-H}}{2}.
\]
\end{prop}

\begin{proof}
We use the same technique as in Theorem~\ref{1step} to derive a lower bound for the terminal risk. We use the notations introduced in \eqref{notation} and consider 

\begin{align}
&\notag\esp{\left(\esp{H}-V_0+\int_0^{t_2}(\theta_{s}-\phi_sX_s) dB_{s}+\eta_{t_1}\deltabr_{t_2}-2G(\eta_{t_1})\Delta t_{2}\right)^2}\\
&\notag=E_G\Bigg[E_G\Bigg[\Big(c+\int_0^{t_2}\varphi_s dB_{s}+\eta_{t_1}\deltabr_{t_2}-2G(\eta_{t_1})\Delta t_{2}\Big)^2\Big|\F_{t_1}\Bigg]\Bigg]\\
&\notag\geq  E_G\Bigg[E_G\Bigg[c+\int_0^{t_2}\varphi_s dB_{s}+\eta_{t_1}\deltabr_{t_2}-2G(\eta_{t_1})\Delta t_{2}\Big|\F_{t_1}\Bigg]^2\vee\\
&\notag\qquad E_G\Bigg[-c-\int_0^{t_2}\varphi_s dB_{s}-\eta_{t_1}\deltabr_{t_2}+2G(\eta_{t_1})\Delta t_{2}\Big|\F_{t_1}\Bigg]^2\Bigg]\\
&\label{x}=E_G\Bigg[\left(c+\int_0^{t_1}\varphi_s dB_{s}\right)^2\vee\Bigg(-c-\int_0^{t_1}\varphi_s dB_{s}+(\sigmau-\sigmad)\Delta t_2|\eta_{t_1}|\Bigg)^2\Bigg],
\end{align}
where we have used that 

\begin{align*}
&\condesp[\F_{t_1}]{c+\int_0^{t_2}\varphi_s dB_{s}+\eta_{t_1}\deltabr_{t_2}-2G(\eta_{t_1})\Delta t_{2}}\\
&\;=\condesp[\F_{t_1}]{c+\int_0^{t_1}\varphi_s dB_{s}+\int_{t_1}^{t_2}\varphi_s dB_{s}+\eta_{t_1}\deltabr_{t_2}-2G(\eta_{t_1})\Delta t_{2}}\\
&\;=c+\int_0^{t_1}\varphi_s dB_{s}
\end{align*}
thanks to Proposition~\ref{prop:useful}, and similarly
\begin{align*}
&\condesp[\F_{t_1}]{-c-\int_0^{t_2}\varphi_s dB_{s}-\eta_{t_1}\deltabr_{t_2}+2G(\eta_{t_1})\Delta t_{2}}\\
&\;=-c-\int_0^{t_1}\varphi_s dB_{s}+(\sigmau-\sigmad)\Delta t_2|\eta_{t_1}|
\end{align*}
as in~\eqref{eq:sp}. This allows us to conclude that the optimal strategy in the interval $(t_1,t_2]$ is given by $\phi^\ast_tX_t=\theta_t$. We now use~\eqref{absvalue1} to rewrite~\eqref{x} as
\begin{equation}
\label{y}
\begin{split}
&E_G\Bigg[\left(c+\int_0^{t_1}\varphi_s dB_{s}\right)^2\vee\Bigg(c-(\sigmau-\sigmad)\Delta t_2\esp{|\eta_{t_1}|}+\\
&\qquad\qquad\qquad\qquad\qquad \int_0^{t_1}\left(\varphi_s-(\sigmau-\sigmad)\Delta t_2\mu_s\right)dB_s\Bigg)^2\Bigg].
\end{split}
\end{equation}
Let us introduce the auxiliary notation 
\begin{equation}
\label{epsilon}
\epsilon:=c-\frac{(\sigmau-\sigmad)\Delta t_2\esp{|\eta_{t_1}|}}{2}
\end{equation} 
and 
\begin{equation}
\label{psi}
\psi_s:=\varphi_s-\frac{(\sigmau-\sigmad)\Delta t_2\mu_s}{2},
\end{equation} 
to further rewrite~\eqref{y} as 
\begin{align*}
\notag&E_G\Bigg[\left(\frac{(\sigmau-\sigmad)\Delta t_2\esp{|\eta_{t_1}|}}{2}+\epsilon+\int_0^{t_1}\left(\frac{(\sigmau-\sigmad)\Delta t_2\mu_s}{2}+\psi_s\right) dB_{s}\right)^2\vee\\
\notag&\Bigg(-\frac{(\sigmau-\sigmad)\Delta t_2\esp{|\eta_{t_1}|}}{2}+\epsilon+\int_0^{t_1}\left(-\frac{(\sigmau-\sigmad)\Delta t_2\mu_s}{2}+\psi_s\right)dB_s\Bigg)^2\Bigg]
\end{align*}
\begin{align*}
\notag=&E_G\Bigg[\Bigg(\frac{(\sigmau-\sigmad)\Delta t_2|\eta_{t_1}|}{2}+\epsilon+\int_0^{t_1}\psi_sdB_{s}\Bigg)^2\vee\\
\notag &\qquad\qquad\qquad \Bigg(-\frac{(\sigmau-\sigmad)\Delta t_2|\eta_{t_1}|}{2}+\epsilon+\int_0^{t_1}\psi_sdB_{s}\Bigg)^2\Bigg]\\
\notag=&E_G\Bigg[\Bigg\{\left(\frac{(\sigmau-\sigmad)\Delta t_2|\eta_{t_1}|}{2}\right)^2+\left(\epsilon+\int_0^{t_1}\psi_sdB_{s}\right)^2+\\
&\qquad\qquad\qquad\qquad\qquad +2\frac{(\sigmau-\sigmad)\Delta t_2|\eta_{t_1}|}{2}\left(\epsilon+\int_0^{t_1}\psi_sdB_{s}\right)\Bigg\}\vee\\
\notag&\Bigg\{\left(\frac{(\sigmau-\sigmad)\Delta t_2|\eta_{t_1}|}{2}\right)^2+\left(\epsilon+\int_0^{t_1}\psi_sdB_{s}\right)^2+\\
&\qquad\qquad\qquad\qquad\qquad -2\frac{(\sigmau-\sigmad)\Delta t_2|\eta_{t_1}|}{2}\left(\epsilon+\int_0^{t_1}\psi_sdB_{s}\right)\Bigg\}\Bigg]\\
\notag=&E_G\Bigg[\Bigg(\frac{(\sigmau-\sigmad)\Delta t_2|\eta_{t_1}|}{2}+\left| \epsilon+\int_0^{t_1}\psi_s dB_s\right|\Bigg)^2\Bigg],
\end{align*}
where in the first equality we used the representation of $|\eta_{t_1}|$ in~\eqref{absvalue1}. The minimum is obtained by setting $\epsilon=0$ and $\psi_t=0$ on $(0,t_1]$.
\end{proof}

\begin{defn}
The parameter $\epsilon$ in~\eqref{epsilon} is called \emph{admissible} if the corresponding value of $V_0$ is such that $V_0\in(-\esp{-H},\esp{H})$. 
\end{defn}

\noindent
In order to solve the second step of our recursive scheme we first introduce the following auxiliary lemmas.

\begin{lemma}
\label{lemmapartizione}
For any $t\in[0,T]$ and any $X\in L_G^p(\F_t)$, with $p\geq1$ there exists a sequence of random variables of the form 
\[
X_n=\sum_{i=0}^{n-1}\mathbb{I}_{A_i}x_i,
\]
where $\{A_i\}_{i=0,\dots,n-1}$ is a partition of $\Omega$, $A_i\in\F_t$ and $x_i\in\R$, such that 
\[
\Vert X-X_n\Vert_{p}\longrightarrow 0, \quad n\to\infty.
\]
\end{lemma}
\begin{proof}
Fix $N,n\in\NM$ and let 
\[
X_n:=\sum_{i=0}^{n-1} \frac{N}{n}i\; \mathbb{I}_{\left\{\frac{N}{n}i\leq |X| <\frac{N}{n}(i+1)\right\}}.
\]
It follows that 
\begin{equation}
\label{discretization}
\begin{split}
\esp{(X-X_n)^p}&=\esp{X^p\mathbb{I}_{\left\{|X|>N\right\}}+\sum_{i=0}^{n-1}(X-\frac{N}{n}i)^p\mathbb{I}_{\left\{\frac{N}{n}i\leq |X| <\frac{N}{n}(i+1)\right\}}}\\
&\leq \esp{X^p\mathbb{I}_{\left\{|X|>N\right\}}}+\esp{\sum_{i=0}^{n-1}(X-\frac{N}{n}i)^p\mathbb{I}_{\left\{\frac{N}{n}i\leq |X| <\frac{N}{n}(i+1)\right\}}}\\
&\leq \esp{X^p\mathbb{I}_{\left\{|X|>N\right\}}}+\left(\frac{N}{n}\right)^p\esp{\mathbb{I}_{ \left\{|X| \leq N\right\}}}.
\end{split}
\end{equation}
Since by Theorem 25 in~\cite{dhp} we have that $\esp{X^p\mathbb{I}_{|X|>N}}$ converges to zero as $N$ tends to infinity, we can conclude by first letting $n\to\infty$ and then $N\to\infty$ in~\eqref{discretization}.
\end{proof}

\begin{lemma}
\label{auxiliary1}
For any $t\leq T$ and $n\in\mathbb{N}$ let $\{A_1,\dots,A_n\}$ be a partition of $\Omega$ such that $A_i\in\F_t$ for every $i\in\{1,\dots,n\}$. It holds that 
\[
\inf_{\psi\in M_G^2(0,t)} E^{P}\left[\sum_{i=1}^n \mathbb{I}_{A_i}\left(x_i+|\epsilon+\int_0^t\psi_s dB_s|\right)^2\right]=E^{P}\left[\sum_{i=1}^n \mathbb{I}_{A_i}\left(x_i+|\epsilon|\right)^2\right],
\]
for every $\epsilon\in\R$, $P\in\Pro$ and $\{x_1,\dots,x_n\}\in \R_+^n$.
\end{lemma}
\begin{proof}
We assume without loss of generality that $\{x_1,\dots,x_n\}$ are all different and increasingly ordered. The result is achieved by induction. If $n=1$ the claim trivially holds. To prove the induction step suppose there exists a $\bar{\psi}\in M_G^2(0,t)$ such that 
\begin{equation}
\label{inductionaux}
E^{P}\left[\sum_{i=1}^{n+1} \mathbb{I}_{A_i}\left(x_i+|\epsilon+\int_0^t\bar{\psi_s} dB_s|\right)^2\right]<E^{P}\left[\sum_{i=1}^{n+1} \mathbb{I}_{A_i}\left(x_i+|\epsilon|\right)^2\right].
\end{equation}
We show that this, together with the induction hypothesis, leads to a contradiction. To this purpose we replace $x_j$, where $j\notin\{1,n+1\}$, with a $x_k$ with $k\in\{1,\dots,n+1\}\setminus {j}$, in order to get a sum of only $n$ different elements and proceed as follows. Note that~\eqref{inductionaux} is equivalent to 
\begin{align}
\notag E^{P}\left[\sum_{i=1}^{n+1} \mathbb{I}_{A_i}\left(\tilde{x}_i+|\epsilon+\int_0^t\bar{\psi_s} dB_s|\right)^2\right]&<E^{P}\left[\sum_{i=1}^{n+1} \mathbb{I}_{A_i}\left(\tilde{x}_i+|\epsilon|\right)^2\right]\\
\notag&\;+E^{P}\left[\mathbb{I}_{A_j}\left(x+|\epsilon+\int_0^t\bar{\psi_s} dB_s|\right)^2\right]\\
\notag&\;-E^{P}\left[\mathbb{I}_{A_j}\left(x_j+|\epsilon+\int_0^t\bar{\psi_s} dB_s|\right)^2\right]\\
\notag &\;+E^{P}\left[\mathbb{I}_{A_j}\left(x_j+|\epsilon|\right)^2\right]\\
\label{inductionaux2}&\;-E^{P}\left[\mathbb{I}_{A_j}\left(x+|\epsilon|\right)^2\right],
\end{align} 
where $x\in\R_+$, and $\{\tilde{x}_1,\dots,\tilde{x}_{n+1}\}$ stands for the new sequence in which $x_j$ has been replaced by $x$. To conclude we consider
\begin{align}
\notag & E^{P}\left[\mathbb{I}_{A_j}\left(x+|\epsilon+\int_0^t\bar{\psi_s} dB_s|\right)^2\right]-E^{P}\left[\mathbb{I}_{A_j}\left(x_j+|\notag \epsilon+\int_0^t\bar{\psi_s} dB_s|\right)^2\right]+\\
\notag &+E^{P}\left[\mathbb{I}_{A_j}\left(x_j+|\epsilon|\right)^2\right]-E^{P}\left[\mathbb{I}_{A_j}\left(x+|\epsilon|\right)^2\right]\\
\notag=& E^{P}\left[\mathbb{I}_{A_j}(x-x_j)\left(x+x_j+2|\epsilon+\int_0^t\bar{\psi_s} dB_s|\right)\right]+\\
\notag&-E^{P}\left[\mathbb{I}_{A_j}(x-x_j)\left(x+x_j+2|\epsilon|\right)\right]\\
\label{aux3}=&E^{P}\left[2\mathbb{I}_{A_j}(x-x_j)\left(|\epsilon+\int_0^t\bar{\psi_s} dB_s|-|\epsilon|\right)\right].
\end{align}
If now
\[
E^{P}\left[\mathbb{I}_{A_j}\left(|\epsilon+\int_0^t\bar{\psi_s} dB_s|-|\epsilon|\right)\right]\geq 0
\]
we choose $x=x_k$ for any $k\in{1,\dots,j-1}$ and obtain for the partition 
\begin{equation}
\label{514}
\begin{split}
\{\tilde{A_1},\dots,\tilde{A_n}\}:=\{A_1,\dots,A_{k-1},A_k\cup A_j,A_{k+1},\dots,A_{j-1},A_{j+1},\dots,A_{n+1}\}
\end{split}
\end{equation} 
and 
\begin{equation}
\label{517star}
\{y_1,\dots,y_n\}:=\{x_1,\dots,x_{k-1},x_k,x_{k+1},\dots,x_{j-1},x_{j+1},\dots,x_{n+1}\}
\end{equation}
that
\begin{align}
\notag &E^{P}\left[\sum_{i=1}^{n} \mathbb{I}_{\tilde{A_i}}\left(y_i+|\epsilon+\int_0^t\bar{\psi_s} dB_s|\right)^2\right]=E^{P}\left[\sum_{i=1}^{n+1} \mathbb{I}_{A_i}\left(\tilde{x_i}+|\epsilon+\int_0^t\bar{\psi_s} dB_s|\right)^2\right]\\
\label{nuo}&<E^{P}\left[\sum_{i=1}^{n+1} \mathbb{I}_{A_i}\left(\tilde{x_i}+|\epsilon|\right)^2\right]=E^{P}\left[\sum_{i=1}^{n} \mathbb{I}_{\tilde{A_i}}\left(y_i+|\epsilon|\right)^2\right],
\end{align}
in contradiction with the induction hypothesis. If
\[
E^{P}\left[\mathbb{I}_{A_j}\left(|\epsilon+\int_0^t\bar{\psi_s} dB_s|-|\epsilon|\right)\right]< 0,
\]
we obtain~\eqref{nuo} with $x=x_k$ for any $k\in{j+1,\dots,n+1}$.
\end{proof}

\begin{lemma}
\label{auxiliary2}
Under the hypothesis of Lemma~\ref{auxiliary1} and for any $\eta_{t_0}\in\R$ it holds that
\begin{align*}
&\esp{\sum_{i=1}^n \mathbb{I}_{A_i}\left(x_i+|\epsilon+\eta_{t_0}\deltabr_{t}-2G(\eta_{t_0})\Delta t|\right)^2} =  \\
&  = \sup_{\overset{\sigma\in \mathcal{A}_{0,t}^\Theta}{\sigma \text{ constant}}} E^{P^\sigma}\left[\sum_{i=1}^n \mathbb{I}_{A_i}\left(x_i+|\epsilon+\eta_{t_0}\deltabr_{t}-2G(\eta_{t_0})\Delta t|\right)^2\right] =\\
& = E^{P^{\sigma^\ast}}\left[\sum_{i=1}^n \mathbb{I}_{A_i}\left(x_i+|\epsilon+\eta_{t_0}\deltabr_{t}-2G(\eta_{t_0})\Delta t|\right)^2\right], 
\end{align*}
for some $\sigma^\ast\in[\underline{\sigma},\overline{\sigma}]$.
\end{lemma}
\begin{proof}
We denote for simplicity 
\[
-K_t:=\eta_{t_0}\deltabr_{t}-2G(\eta_{t_0})\Delta t,
\]
and proceed again by induction, using the same conventions as in Lemma~\ref{auxiliary1}. In particular, also here we assume that $\{x_1,\dots,x_n\}$ are all different and increasingly ordered. The case $n=1$ is clear because of~\eqref{maximallydist}, as $\deltabr_{t}$ is maximally distributed. Assume now there exists a $P\in\Pro$, which is not in the set $\{P^\sigma,\;\sigma\in[\underline{\sigma},\overline{\sigma}],\;\sigma \text{ constant}\}$, such that 
\begin{equation}
\label{aux4}
 E^{P}\left[\sum_{i=1}^{n+1} \mathbb{I}_{A_i}\left(x_i+|\epsilon-K_t|\right)^2\right] > E^{P^{\sigma^\ast}}\left[\sum_{i=1}^{n+1} \mathbb{I}_{A_i}\left(x_i+|\epsilon-K_t|\right)^2\right]. 
\end{equation}
The expression~\eqref{aux4} implies that there exists a $j\in \{1,\dots,n+1\}$ such that
\begin{equation}
\label{aux6}
 E^{P}\left[\mathbb{I}_{A_j}\left(x_j+|\epsilon-K_t|\right)^2\right] > E^{P^{\sigma^\ast}}\left[\mathbb{I}_{A_j}\left(x_j+|\epsilon-K_t|\right)^2\right],
\end{equation}
which is equivalent to 
\begin{equation}
\label{aux5}
\begin{split}
&\left(P(A_j)-P^{\sigma^\ast}(A_j)\right)x_j^2+2x_j\left(E^{P}\left[\mathbb{I}_{A_j}|\epsilon-K_t|\right]-E^{P^{\sigma^\ast}}\left[\mathbb{I}_{A_j}|\epsilon-K_t|\right]\right)+\\
&\qquad\qquad +E^{P}\left[\mathbb{I}_{A_j}|\epsilon-K_t|^2\right]-E^{P^{\sigma^\ast}}\left[\mathbb{I}_{A_j}|\epsilon-K_t|^2\right]>0.
\end{split}
\end{equation}
Note that, in order for~\eqref{aux6} to hold, we must have $P(A_j)-P^{\sigma^\ast}(A_j)>0$. This implies that~\eqref{aux5} is a convex function in $x_j$, which tends to infinity as $x_j$ tends to infinity. As in Lemma~\ref{auxiliary1}, we get to a contradiction by reducing~\eqref{aux4} to a sum of only $n$ different terms, by replacing $x_j$ with another suitable value. We note that~\eqref{aux4} is equivalent to
\begin{equation}
\label{aux7}
\begin{split}
 E^{P}\left[\sum_{i=1}^{n+1} \mathbb{I}_{A_i}\left(\tilde{x}_i+|\epsilon-K_t|\right)^2\right] >& E^{P^{\sigma^\ast}}\left[\sum_{i=1}^{n+1} \mathbb{I}_{A_i}\left(\tilde{x}_i+|\epsilon-K_t|\right)^2\right]\\
 &+E^{P^{\sigma^\ast}}\left[\mathbb{I}_{A_j}\left(x_j+|\epsilon-K_t|\right)^2\right]\\
 &-E^{P^{\sigma^\ast}}\left[\mathbb{I}_{A_j}\left(x+|\epsilon-K_t|\right)^2\right]\\
 &+E^{P}\left[\mathbb{I}_{A_j}\left(x+|\epsilon-K_t|\right)^2\right]\\
 &-E^{P}\left[\mathbb{I}_{A_j}\left(x_j+|\epsilon-K_t|\right)^2\right],\\
\end{split}  
\end{equation}
where $x\in\R$ and $\{\tilde{x}_1,\dots,\tilde{x}_{n+1}\}$ stands for the new sequence in which $x_j$ has been replaced by $x$ as in Lemma~\ref{auxiliary1}. To conclude, we consider
\[
\begin{split}
&E^{P^{\sigma^\ast}}\left[\mathbb{I}_{A_j}\left(x_j+|\epsilon-K_t|\right)^2\right]-E^{P^{\sigma^\ast}}\left[\mathbb{I}_{A_j}\left(x+|\epsilon-K_t|\right)^2\right]>\\
&E^{P}\left[\mathbb{I}_{A_j}\left(x_j+|\epsilon-K_t|\right)^2\right]-E^{P}\left[\mathbb{I}_{A_j}\left(x+|\epsilon-K_t|\right)^2\right],\\
\end{split}
\]
which is equivalent to 
\begin{equation}
\label{aux8}
\begin{split}
&E^{P^{\sigma^\ast}}\left[\mathbb{I}_{A_j}(x_j-x)\left(x_j+x+2|\epsilon-K_t|\right)\right]>\\
&\;E^{P}\left[\mathbb{I}_{A_j}(x_j-x)\left(x_j+x+2|\epsilon-K_t|\right)\right].
\end{split}
\end{equation}
If $x>x_j$, ~\eqref{aux8} is satisfied if 
\[
E^{P^{\sigma^\ast}}\left[\mathbb{I}_{A_j}\left(\frac{x_j+x}{2}+|\epsilon-K_t|\right)\right]<E^{P}\left[\mathbb{I}_{A_j}\left(\frac{x_j+x}{2}+|\epsilon-K_t|\right)\right],
\]
which in turn is the same as
\begin{equation}
\label{aux9}
\left(P(A_j)-P^{\sigma^\ast}(A_j)\right)\frac{x_j+x}{2}>E^{P^{\sigma^\ast}}\left[\mathbb{I}_{A_j}|\epsilon-K_t|\right]-E^{P}\left[\mathbb{I}_{A_j}|\epsilon-K_t|\right].
\end{equation}
At this point, if there exists a $x=x_k$ satisfying~\eqref{aux9}, where $k\in\{j+1,\dots,n+1\}$, the proof is concluded, as we will get
\[
\begin{split}
E^{P}\left[\sum_{i=1}^{n} \mathbb{I}_{\tilde{A}_i}\left(y_i+|\epsilon-K_t|\right)^2\right] &= E^{P}\left[\sum_{i=1}^{n+1} \mathbb{I}_{A_i}\left(\tilde{x_i}+|\epsilon-K_t|\right)^2\right] \\
&> E^{P^{\sigma^\ast}}\left[\sum_{i=1}^{n+1} \mathbb{I}_{A_i}\left(\tilde{x_i}+|\epsilon-K_t|\right)^2\right]\\
&=E^{P^{\sigma^\ast}}\left[\sum_{i=1}^{n} \mathbb{I}_{\tilde{A}_i}\left(y_i+|\epsilon-K_t|\right)^2\right],
\end{split}
\]
where $\{\tilde{A}_i\}_{i=1,\dots,n}$ and $\{y_i\}_{i=1,\dots,n}$ are introduced in~\eqref{514} and~\eqref{517star}, respectively. If such $x_k$ does not exist, which happens if $j=n+1$ for example, we first substitute some $x_i$ with a $x_r$, where $i\neq r$ and $i,r\in\{1,\dots,n+1\}\setminus{j}$, as in~\eqref{aux7}, and then we substitute $x_j$ with an $x$ sufficiently large to satisfy 
\begin{equation}
\label{auxxxx}
\begin{split}
&E^{P^{\sigma^\ast}}\left[\mathbb{I}_{A_j}\left(x_j+|\epsilon-K_t|\right)^2\right]-E^{P^{\sigma^\ast}}\left[\mathbb{I}_{A_j}\left(x+|\epsilon-K_t|\right)^2\right]\\
&+E^{P}\left[\mathbb{I}_{A_j}\left(x+|\epsilon-K_t|\right)^2\right]-E^{P}\left[\mathbb{I}_{A_j}\left(x_j+|\epsilon-K_t|\right)^2\right]\\
&+E^{P^{\sigma^\ast}}\left[\mathbb{I}_{A_i}\left(x_i+|\epsilon-K_t|\right)^2\right]-E^{P^{\sigma^\ast}}\left[\mathbb{I}_{A_i}\left(x_r+|\epsilon-K_t|\right)^2\right]\\
&+E^{P}\left[\mathbb{I}_{A_i}\left(x_r+|\epsilon-K_t|\right)^2\right]-E^{P}\left[\mathbb{I}_{A_i}\left(x_i+|\epsilon-K_t|\right)^2\right]>0.\\
\end{split}
\end{equation}
This is possible because 
\begin{equation*}
\begin{split}
&E^{P^{\sigma^\ast}}\left[\mathbb{I}_{A_j}\left(x_j+|\epsilon-K_t|\right)^2\right]-E^{P^{\sigma^\ast}}\left[\mathbb{I}_{A_j}\left(x+|\epsilon-K_t|\right)^2\right]\\
&+E^{P}\left[\mathbb{I}_{A_j}\left(x+|\epsilon-K_t|\right)^2\right]-E^{P}\left[\mathbb{I}_{A_j}\left(x_j+|\epsilon-K_t|\right)^2\right]>0
\end{split}
\end{equation*}
is equivalent to~\eqref{aux9}, and its value can be made large enough to ensure~\eqref{auxxxx} because of~\eqref{aux5}.
\end{proof}

\noindent
We can now state the main result.
\begin{teo}
\label{final2}
Consider a claim $H$ of the form 
\[
H=\esp{H}+\int_0^{t_2}\theta_{s} dB_{s}+\eta_{t_0}\deltabr_{t_1}-2G(\eta_{t_0})\Delta t_{1}+\eta_{t_1}\deltabr_{t_2}-2G(\eta_{t_1})\Delta t_{2},
\]
where $0= t_{0}<t_1<t_{2}=T$, $(\theta_{s})_{s\in[0,t_2]}\in M_G^2(0,t_2)$, $\eta_{t_0}\in\R$, $\eta_{t_1}\in L_G^2(\F_{t_1})$ and 
\begin{equation}
\label{absvalue2}
|\eta_{t_1}|=\esp{|\eta_{t_1}|}+\int_0^{t_1}\mu_{s} dB_{s},
\end{equation}
for a certain process $(\mu_{s})_{s\in[0,t_1]}\in M_G^2(0,t_1) $. The optimal mean-variance portfolio is given by 
\[
\phi_t^\ast X_t=\left(\theta_{t}-\frac{\mu_{t}(\sigmau-\sigmad)\Delta t_2}{2}\right)\mathbb{I}_{(t_0,t_1]}(t)+\theta_{t}\mathbb{I}_{(t_1,t_2]}(t) 
\]
for $t\in[0,T]$ and 
\[
V_0^\ast=\esp{H}-\frac{1}{2}(\sigmau-\sigmad)\Delta t_2\esp{|\eta_{t_1}|}-\epsilon,
\]
where $\epsilon\in\R$ solves
\[
\inf_{\epsilon} \esp{\left(\frac{|\eta_{t_1}|}{2}(\sigmau-\sigmad)\Delta t_2+|\epsilon+\eta_{t_0}\deltabr_{t_1}-2G(\eta_{t_0})\Delta t_{1}|\right)^2}.
\]
\end{teo}

\begin{proof}
By the same argument as in Proposition~\ref{final1} we conclude that 
\[
\phi_s^\ast X_s=\theta_s \quad \forall\; s\in(t_1,t_2]
\]
and focus on the following expression
\begin{equation}
\label{tre}
\inf_{\epsilon,\psi} \esp{\left(\frac{|\eta_{t_1}|}{2}(\sigmau-\sigmad)\Delta t_2+|\epsilon+\int_0^{t_1}\psi_sdB_s +\eta_{t_0}\deltabr_{t_1}-2G(\eta_{t_0})\Delta t_{1}|\right)^2},
\end{equation}
where $\epsilon$ and $\psi$ are as in~\eqref{epsilon} and~\eqref{psi}. Let $(Y_n)_{n\in\NM}$ be a sequence of random variables approximating $\frac{|\eta_{t_1}|}{2}(\sigmau-\sigmad)\Delta t_2$ in $L_G^2(\F_{t_1})$ as in Lemma~\ref{lemmapartizione}, with $Y_n=\sum_{i=0}^{n-1}\mathbb{I}_{A_{i,n}}y_{i,n}$, $n\in\mathbb{N}$, where $\{A_{i,n}\}_{i=0,\dots,n-1}$ is a partition of $\Omega$, $A_{i,n}\in\F_t$ and $y_{i,n}\in\R_+$. Consider now the auxiliary problem 
\[
\inf_{\epsilon,\psi} \esp{\left(Y_n+|\epsilon+\int_0^{t_1}\psi_sdB_s +\eta_{t_0}\deltabr_{t_1}-2G(\eta_{t_0})\Delta t_{1}|\right)^2}.
\] 
For every $n\in\NM$ and any admissible $\epsilon$ we can derive the following inequalities
\begin{align}
\notag& \esp{\left(Y_n+|\epsilon+\int_0^{t_1}\psi_sdB_s +\eta_{t_0}\deltabr_{t_1}-2G(\eta_{t_0})\Delta t_{1}|\right)^2}\geq \\
\notag\geq & \sup_{\sigma\in[\underline{\sigma}, \overline{\sigma}]} E^{P^\sigma}\left[\left(Y_n+|\epsilon+\int_0^{t_1}\psi_sdB_s +\eta_{t_0}\deltabr_{t_1}-2G(\eta_{t_0})\Delta t_{1}|\right)^2\right]\\
\label{sigmaconst}\geq & \sup_{\sigma\in[\underline{\sigma}, \overline{\sigma}]} E^{P^\sigma}\left[\left(Y_n+|\epsilon +\eta_{t_0}\deltabr_{t_1}-2G(\eta_{t_0})\Delta t_{1}|\right)^2\right]\\
\label{aux10}=& \esp{\left(Y_n+|\epsilon +\eta_{t_0}\deltabr_{t_1}-2G(\eta_{t_0})\Delta t_{1}|\right)^2}.
\end{align}
The inequality~\eqref{sigmaconst} is clear thanks to Lemma~\ref{auxiliary1}, because
\[
\epsilon^{P^\sigma}:=\epsilon+\eta_{t_0}\deltabr_{t_1}-2G(\eta_{t_0})\Delta t_{1}
\]
is constant $P^\sigma$-a.s.\ for every $\sigma\in[\underline{\sigma}, \overline{\sigma}]$ since 
\[
\deltabr_{t_1}=\sigma^2\Delta t_1 \quad P^\sigma\text{-a.s.}
\]
and $y_{i,n}\in\R_+$  $\forall n,i$. The equality~\eqref{aux10} comes directly from Lemma~\ref{auxiliary2}. Hence we can conclude that, for every $n\in\NM$ and any admissible $\epsilon$,
\begin{equation}
\label{due}
\begin{split}
 &\esp{\left(Y_n+|\epsilon+\int_0^{t_1}\psi_sdB_s +\eta_{t_0}\deltabr_{t_1}-2G(\eta_{t_0})\Delta t_{1}|\right)^2}\\
\geq& \;\esp{\left(Y_n+|\epsilon+\eta_{t_0}\deltabr_{t_1}-2G(\eta_{t_0})\Delta t_{1}|\right)^2}.
\end{split}
\end{equation}
By~\eqref{due} we derive by letting $n\to\infty$ that 
\[
\begin{split}
 &\esp{\left(\frac{|\eta_{t_1}|}{2}(\sigmau-\sigmad)\Delta t_2+|\epsilon+\int_0^{t_1}\psi_sdB_s +\eta_{t_0}\deltabr_{t_1}-2G(\eta_{t_0})\Delta t_{1}|\right)^2}\\
\geq& \;\esp{\left(\frac{|\eta_{t_1}|}{2}(\sigmau-\sigmad)\Delta t_2+|\epsilon+\eta_{t_0}\deltabr_{t_1}-2G(\eta_{t_0})\Delta t_{1}|\right)^2},
\end{split}
\]
for any admissible $\epsilon$ and any $\psi\in M_G^2(0,t_1)$, because of the $L_G^2$-convergence of $Y_n$ to $\frac{|\eta_{t_1}|}{2}(\sigmau-\sigmad)\Delta t_2$. This in turn implies
\[
\begin{split}
 &\inf_{\epsilon,\psi}\esp{\left(\frac{|\eta_{t_1}|}{2}(\sigmau-\sigmad)\Delta t_2+|\epsilon+\int_0^{t_1}\psi_sdB_s +\eta_{t_0}\deltabr_{t_1}-2G(\eta_{t_0})\Delta t_{1}|\right)^2}\\
\geq& \;\inf_{\epsilon}\esp{\left(\frac{|\eta_{t_1}|}{2}(\sigmau-\sigmad)\Delta t_2+|\epsilon+\eta_{t_0}\deltabr_{t_1}-2G(\eta_{t_0})\Delta t_{1}|\right)^2}.
\end{split}
\]
\end{proof}

\noindent
As a particular example, we get now the expression of the mean-variance optimal portfolio for a particular claim of the type introduced in Theorem~\ref{final2}, for which we are able to determine explicitly also the optimal initial wealth $V_0^\ast$.
\begin{esempio}
Consider a claim $H$ of the following form 
\[
H=\esp{H}+\int_0^{t_2}\theta_{s} dB_{s}+\eta_{t_0}\deltabr_{t_1}-2G(\eta_{t_0})\Delta t_{1}+\eta_{t_1}\deltabr_{t_2}-2G(\eta_{t_1})\Delta t_{2},
\]
where $0= t_{0}<t_1<t_{2}=T$, $(\theta_{s})_{s\in[0,t_2]}\in M_G^2(0,t_2)$, $\eta_{t_0}\in\R_+$, $\eta_{t_1}\in L_G^2(\F_{t_1})$ and 
\begin{equation}
|\eta_{t_1}|=\exp\left(B_{t_1}-\frac{1}{2}\brac_{t_1}\right)=1+\int_0^{t_1}e^{B_{s}-\frac{1}{2}\brac_{s}} dB_{s}.
\end{equation}
Assume moreover that 
\begin{equation}
\label{speriamo2}
\frac{1}{2}\Delta t_2 e^{\frac{1}{2}\sigmau \Delta t_1}\geq \eta_{t_0} \Delta t_1+\frac{1}{2}\Delta t_2.
\end{equation}
The optimal mean-variance portfolio is given by
\[
X_t\phi_t^\ast=\left(\theta_{t}-\frac{e^{B_{t}-\frac{1}{2}\brac_{t}}(\sigmau-\sigmad)\Delta t_2}{2}\right)\mathbb{I}_{(t_0,t_1]}(t)+\theta_{t}\mathbb{I}_{(t_1,t_2]}(t)
\]
for $t\in[0,T]$ and 
\begin{equation}
\label{528}
V_0^\ast=\esp{H}-\frac{(\sigmau-\sigmad)\Delta t_2 }{2}.
\end{equation}
\end{esempio}
\begin{proof}
By Theorem~\ref{final2} we only have to find the infimum of
\begin{equation}
\label{speriamo}
\begin{split}
 &E_G\Bigg[\Bigg(\frac{1}{2}(\sigmau-\sigmad)\Delta t_2 e^{B_{t_1}-\frac{1}{2}\brac_{t_1}} +\left|\epsilon+\eta_{t_0}\deltabr_{t_1}-2G(\eta_{t_0})\Delta t_1\right|\Bigg)^2\Bigg].
\end{split}
\end{equation}
As the expression~\eqref{speriamo} is always bigger than 
\[
\begin{split}
\esp{\left(\frac{1}{2}(\sigmau-\sigmad)\Delta t_2 e^{B_{t_1}-\frac{1}{2}\brac_{t_1}}\right)^2}&=\frac{1}{4}(\sigmau-\sigmad)^2\Delta t_2^2\esp{e^{2B_{t_1}-\brac_{t_1}}}\\
&=\frac{1}{4}(\sigmau-\sigmad)^2\Delta t_2^2 E^{P^{\overline{\sigma}}}\left[e^{2B_{t_1}-\brac_{t_1}}\right]\\
&=\frac{1}{4}(\sigmau-\sigmad)^2\Delta t_2^2 e^{\sigmau \Delta t_1},
\end{split}
\]
we prove~\eqref{528} by showing that with the particular choice $\epsilon=0$ the quantity~\eqref{speriamo} reaches this lower bound.  To this end one has to prove that
\[
\begin{split}
&\sup_{\sigma\in\mathcal{A}_{0,t_1}^\Theta} E^P\left[
\left(\frac{1}{2}(\sigmau-\sigmad)\Delta t_2 e^{\int_0^{t_1}\sigma_sdW_s-\frac{1}{2}\int_0^{t_1}\sigma_s^2ds}+\eta_{t_0}|\int_0^{t_1} 
(\sigma_s^2-\sigmau)ds|\right)^2\right]\\
=&\esp{\left(\frac{1}{2}(\sigmau-\sigmad)\Delta t_2 e^{B_{t_1}-\frac{1}{2}\brac_{t_1}}+|\eta_{t_0}\deltabr_{t_1}-2G(\eta_{t_0})\Delta t_1|\right)^2}\\
=&\esp{\left(\frac{1}{2}(\sigmau-\sigmad)\Delta t_2 e^{B_{t_1}-\frac{1}{2}\brac_{t_1}}\right)^2},
\end{split}
\]
where $\mathcal{A}_{0,t_1}^\Theta$ denotes the set of $\mathbb{F}$-adapted processes on $[0,t_1]$ taking values in $[\underline{\sigma},\overline{\sigma}]$. This holds if the inequality
\begin{equation}
\label{aux99}
\begin{split}
&E^P\left[
\left(\frac{1}{2}(\sigmau-\sigmad)\Delta t_2 e^{\int_0^{t_1}\sigma_sdW_s-\frac{1}{2}\int_0^{t_1}\sigma_s^2ds}+\eta_{t_0}\int_0^{t_1} 
(\sigmau-\sigma_s^2)ds\right)^2\right]\\
&\qquad\qquad\leq\frac{1}{4}(\sigmau-\sigmad)^2\Delta t_2^2 e^{\sigmau \Delta t_1}
\end{split}
\end{equation}
is verified for any $\sigma\in\mathcal{A}_{0,t_1}^\Theta$. As~\eqref{aux99} holds if and only if we have
\[
\begin{split}
&E^P\Bigg[\Bigg(\frac{1}{2}(\sigmau-\sigmad)\Delta t_2 \Big(e^{\int_0^{t_1}\sigma_sdW_s-\frac{1}{2}\int_0^{t_1}\sigma_s^2ds}+e^{\frac{1}{2}\sigmau \Delta t_1}\Big)+\eta_{t_0}\int_0^{t_1}
(\sigmau-\sigma_s^2)ds \Bigg)\\
&\quad\cdot\Bigg(\frac{1}{2}(\sigmau-\sigmad)\Delta t_2 \Big(e^{\int_0^{t_1}\sigma_sdW_s-\frac{1}{2}\int_0^{t_1}\sigma_s^2ds}-e^{\frac{1}{2}\sigmau \Delta t_1}\Big)+\eta_{t_0}\int_0^{t_1}
(\sigmau-\sigma_s^2)ds \Bigg)\Bigg]\leq 0,
\end{split}
\]
we complete the proof by showing that the previous expression is bounded from above by
\[
\begin{split}
& \lim_{N\to\infty}  C(N)\; E^P\Bigg[\Bigg(\frac{1}{2}(\sigmau-\sigmad)\Delta t_2 \Big(e^{\int_0^{t_1}\sigma_sdW_s-\frac{1}{2}\int_0^{t_1}\sigma_s^2ds}-e^{\frac{1}{2}\sigmau \Delta t_1}\Big)\\
&\qquad\qquad\qquad\qquad\qquad\qquad +\eta_{t_0}\int_0^{t_1}(\sigmau-\sigma_s^2)ds \Bigg)\mathbb{I}_{\{\int_0^{t_1}\sigma_sdW_s<N\}}\Bigg]\\
\leq &\lim_{N\to\infty}  C(N)\left(\frac{1}{2}(\sigmau-\sigmad)\Delta t_2 \left(1-e^{\frac{1}{2}\sigmau \Delta t_1}\right)+
\eta_{t_0} E^P\left[\int_0^{t_1}(\sigmau-\sigma_s^2)ds\right]\right)\\
\leq &\lim_{N\to\infty}  C(N)\left(\frac{1}{2}(\sigmau-\sigmad)\Delta t_2 \left(1-e^{\frac{1}{2}\sigmau \Delta t_1}\right)+
\eta_{t_0} (\sigmau-\sigmad)\right)<0,
\end{split}
\]
where the last inequality comes from condition~\eqref{speriamo2} and $C(N)$ is a positive constant for each $N\in\mathbb{N}$.
\end{proof}

\noindent
It is quite straightforward to extend the result of Theorem~\ref{final2} by generalizing the decomposition of $|\eta_{t_1}|$, and thus completing the second step of our scheme.
\begin{teo}
Consider a claim $H$ of the form 
\[
H=\esp{H}+\int_0^{t_2}\theta_{s} dB_{s}+\eta_{t_0}\deltabr_{t_1}-2G(\eta_{t_0})\Delta t_{1}+\eta_{t_1}\deltabr_{t_2}-2G(\eta_{t_1})\Delta t_{2},
\]
where $0= t_{0}<t_1<t_{2}=T$, $(\theta_{s})_{s\in[0,t_2]}\in M_G^2(0,t_2)$, $\eta_{t_0}\in\R$, $\eta_{t_1}\in L_G^2(\F_{t_1})$ and 
\begin{equation}
|\eta_{t_1}|=\esp{|\eta_{t_1}|}+\int_0^{t_1}\mu_{s} dB_{s}+\xi_{t_0}\deltabr_{t_1}-2G(\xi_{t_0})\Delta t_{1},
\end{equation}
for a certain process $(\mu_{s})_{s\in[0,t_1]}\in M_G^2(0,t_1) $ and $\xi_{t_0}\in\R$. The optimal mean-variance portfolio is given by 
\[
\phi_t^\ast X_t=\left(\theta_{t}-\frac{\mu_{t}(\sigmau-\sigmad)\Delta t_2}{2}\right)\mathbb{I}_{(t_0,t_1]}(t)+\theta_{t}\mathbb{I}_{(t_1,t_2]}(t) 
\]
and 
\[
V_0^\ast=\esp{H}-\frac{1}{2}(\sigmau-\sigmad)\Delta t_2\esp{|\eta_{t_1}|}-\epsilon,
\]
where $\epsilon\in\R$ solves
\[
\begin{split}
&\inf_{\epsilon} E_G\Bigg[\Bigg(\frac{|\eta_{t_1}|}{2}(\sigmau-\sigmad)\Delta t_2+\Big|\epsilon+\left(\eta_{t_0}-\frac12(\sigmau-\sigmad)\xi_{t_0}\Delta t_1\right)\deltabr_{t_1}+\\
&\qquad\qquad\qquad-2\left(G(\eta_{t_0})-\frac12(\sigmau-\sigmad)\Delta t_1G(\xi_{t_0})\right)\Delta t_{1}\Big|\Bigg)^2\Bigg].
\end{split}
\]
\end{teo}
\begin{proof}
The proof follows the same steps as in Theorem~\ref{final2} and is omitted.
\end{proof}

\section{Bounds for the Terminal Risk}
The extension to the general piecewise constant case is much more involved. It is clear however, given the explicit achievements of Section 4, that in order to obtain a general result it is crucial to study the mean-variance problem in the situation where 
\[
|\eta_t|=|\eta_0|+\int_0^t{\mu_s}dB_s,
\]
for every $t\in[0,T]$, with $(\mu_t)\ut\in M_G^2[0,T]$. As a partial answer to this issue we provide here a lower and upper bound for the optimal terminal risk.

\begin{lemma}
Consider a claim $H$ of the form 
\[
H=\esp{H}+\int_0^{T}\theta_{s} dB_{s}+\int_0^T\eta_s d\langle B\rangle_s-2\int_0^TG(\eta_s)ds,
\]
where $(\theta_{s})_{s\in[0,T]}\in M_G^2(0,T)$, $(\eta_{s})_{s\in[0,T]}\in M_G^1(0,T)$ and 
\[
|\eta_t|=|\eta_0|+\int_0^t{\mu_s}dB_s,
\]
for a certain process $(\mu_{s})_{s\in[0,T]}\in M_G^2(0,T)$, for every $t\in[0,T]$. The optimal terminal risk~\eqref{problemadef} lies in the closed interval $[\underline{J}(V_0,\phi),\overline{J}(V_0,\phi)]$, where 
\[
\begin{split}
 \underline{J}(V_0,\phi)&=\left(\frac{\esp{-\int_0^T\eta_sd\langle B\rangle_s+2\int_0^TG(\eta_s)ds}}{2}\right)^2=\left(\frac{\esp{K_T}}{2}\right)^2,\\
 \overline{J}(V_0,\phi)&=\esp{\left(\frac{(\sigmau-\sigmad)}{2}\int_0^T|\eta_s|ds\right)^2}.
\end{split} 
\]
\end{lemma}

\begin{proof}
We start with the computation of the upper bound for $J(V_0,\phi)$:
\begin{align}
\notag&\esp{\left(\esp{H}-V_0+\int_0^T\left(\theta_s-\phi_sX_s\right)dB_s+\int_0^T\eta_sd\langle B\rangle_s-2\int_0^TG(\eta_s)ds\right)^2}\\
&\notag\leq E_G\Bigg[ \left(\esp{H}-V_0+\int_0^T\left(\theta_s-\phi_sX_s\right)dB_s\right)^2\vee\\
\label{neqcomp1}&\qquad\qquad\left(\esp{H}-V_0+\int_0^T\left(\theta_s-\phi_sX_s\right)dB_s-(\sigmau-\sigmad)\int_0^T|\eta|_sds\right)^2\Bigg]
\end{align}
\begin{align}
&\notag=E_G\Bigg[ \left(\esp{H}-V_0+\int_0^T\left(\theta_s-\phi_sX_s\right)dB_s\right)^2\vee\Big(\esp{H}-V_0+\\
\label{neqcomp2}
&\quad\quad-|\eta_0|(\sigmau-\sigmad)T+\int_0^T\left(\theta_s-\phi_sX_s-(T-s)(\sigmau-\sigmad)\mu_s\right)dB_s\Big)^2\Bigg],
\end{align} 
where we used that 
\[
\int_0^T\eta_sd\langle B\rangle_s-2\int_0^TG(\eta_s)ds\;\in\;[-(\sigmau-\sigmad)\int_0^T|\eta|_sds,0]
\]
in~\eqref{neqcomp1} and that 
\[
\begin{split}
\int_0^T|\eta|_sds&=\int_0^T\left(|\eta_0|+\int_0^s\mu_udB_u\right)ds\\
&=|\eta_0|T+\int_0^T\int_0^s\mu_udB_uds\\
&=|\eta_0|T+T\int_0^T\mu_sdB_s-\int_0^Ts\mu_sdB_s\\
&=|\eta_0|T+\int_0^T(T-s)\mu_sdB_s
\end{split}
\]
in~\eqref{neqcomp2}. We now perform the same change of variables seen in Proposition~\ref{final1} by setting
\[
\begin{split}
\epsilon&:=\esp{H}-V_0-\frac{T}{2}(\sigmau-\sigmad)|\eta_0|,\\
\psi_t&:=	\theta_t-\phi_tX_t-\frac{(T-s)}{2}(\sigmau-\sigmad)\mu_t,
\end{split}
\] 
to rewrite~\eqref{neqcomp2} as
\[
\begin{split}
&\esp{\left(\frac{T}{2}(\sigmau-\sigmad)|\eta_0|+\int_0^T\frac{(T-s)}{2}(\sigmau-\sigmad)\mu_sdB_s+\left| \epsilon+\int_0^T\psi_sdB_s\right|\right)^2}\\
=&\esp{\left(\frac{(\sigmau-\sigmad)}{2}\int_0^T|\eta|_sds+\left| \epsilon+\int_0^T\psi_sdB_s\right|\right)^2}
\end{split}
\]
which is minimal when $\epsilon=0$ and $\psi\equiv 0$, see also the proof of Proposition~\ref{final1}. On the other hand a lower bound is obtained by means of the $G$-Jensen inequality. As in Theorem~\ref{1step} we get the following chain of inequalities
\begin{align}
\notag
&\esp{\left(\esp{H}-V_0+\int_0^T\left(\theta_s-\phi_sX_s\right)dB_s+\int_0^T\eta_sd\langle B\rangle_s-2\int_0^TG(\eta_s)ds\right)^2}\\
\label{extra3}
\geq &\left(\esp{H}-V_0\right)^2 \vee \left(\esp{H}-V_0+\esp{-\int_0^T\eta_sd\langle B\rangle_s+2\int_0^TG(\eta_s)ds}\right)^2
\\
\label{extra4}
\geq & \left(\frac{\esp{-\int_0^T\eta_sd\langle B\rangle_s+2\int_0^TG(\eta_s)ds}}{2}\right)^2,
\end{align}
where we have used Proposition~\ref{prop:useful} in~\eqref{extra3} and chosen  
\[
\bar{V}_0=\esp{H}-\frac{\esp{-\int_0^T\eta_sd\langle B\rangle_s+2\int_0^TG(\eta_s)ds}}{2}
\]
to minimize the expression over $V_0$ and obtain~\eqref{extra4}.
\end{proof}

\end{document}